\documentclass[a4paper,UKenglish]{lipics}

\usepackage{amssymb}
\usepackage{epsfig}
\usepackage{xspace}
\usepackage{color}
\usepackage{wrapfig}
\usepackage{microtype}
\usepackage{mathtools}
\usepackage{amsmath}

\usepackage{graphicx}
\usepackage{caption}
\usepackage{float}
\usepackage{wrapfig}

\graphicspath{{IpeFigures/}{Figures/}}

\newcommand{\denselist}{\itemsep -2pt\parsep=-1pt\partopsep -2pt}

\newcommand{\E}{\ensuremath{\mathsf{E}}}

\newcommand{\eps}{\varepsilon}
\newcommand{\poly}{\ensuremath{\mathsf{poly}}}

\renewcommand{\c}[1]{\ensuremath{\mathcal{#1}}}
\newcommand{\s}[1]{\textsf{\small #1}}

\newcommand{\R}{\mathbb{R}}

\newcommand{\omt}[1]{}
\newcommand{\etal}{\emph{et al.}\xspace}

\renewcommand{\c}[1]{\ensuremath{\mathcal{#1}}}

\newcommand{\SatScan}{\s{SatScan}\xspace}
\newcommand{\netScan}{\s{netScan}\xspace}
\newcommand{\gridScan}{\s{gridScan}\xspace}
\newcommand{\gridScanLin}{\s{gridScan\_linear}\xspace}

\title{Computing Approximate Statistical Discrepancy}

\author[1]{Michael Matheny}
\author[1]{Jeff M. Phillips}
\affil[1]{University of Utah\\
	\texttt{\{mmath|jeffp\}@cs.utah.edu }}
\authorrunning{M.\, Matheny and J.\,M. Phillips} 

\Copyright{Michael Matheny and Jeff Phillips}

\subjclass{Theory of computation $\sim$ Computational geometry}

\serieslogo{}
\volumeinfo
{Billy Editor and Bill Editors}
{2}
{}
{1}
{1}
{1}
\EventShortName{}

\volumeinfo
{}
{0}
{}
{}
{}
{1}

\begin{document}

\maketitle

\begin{abstract}
Consider a geometric range space $(X,\c{A})$ where $X$ is comprised of the union of a red set $R$ and blue set $B$.  Let $\Phi(A)$ define the absolute difference between the fraction of red and fraction of blue points which fall in the range $A$.  
The maximum discrepancy range $A^* = \arg \max_{A \in (X,\c{A})} \Phi(A)$.  
Our goal is to find some $\hat{A} \in (X,\c{A})$ such that 
$\Phi(A^*) - \Phi(\hat A) \leq \eps$.  
We develop general algorithms for this approximation problem for range spaces with bounded VC-dimension, as well as significant improvements for specific geometric range spaces defined by balls, halfspaces, and axis-aligned rectangles.  This problem has direct applications in discrepancy evaluation and classification, and we also show an improved reduction to a class of problems in spatial scan statistics.  
\footnote{Thanks to support by NSF CCF-1350888, IIS-1251019, ACI-1443046, CNS-1514520, and CNS-1564287}
\end{abstract}

\section{Introduction}

Let $X$ be a set of $m$ points in $\R^d$ for constant $d$.  Let $X = R \cup B$ be the union (possibly not disjoint) of two sets $R$, the red set, and $B$, the blue set.  
Also consider an associated range space $(X, \c{A})$; we are particularly interested in range spaces defined by geometric shapes such as rectangles in $\R^d$ $(X, \c{R}_d)$, disks in $\R^2$ $(X, \c{D})$, and $d$-dimensional halfspaces $(X, \c{H}_d)$.  

Let $\mu_R(A) = |R \cap A|/|R|$ and $\mu_B(A) = |B \cap A|/|B|$ be the fraction of red or blue points, respectively, in the range $A$.  
We study the discrepancy function $\Phi_X(A) = |\mu_R(A) - \mu_B(A)|$, when for brevity is typically write as just $\Phi(A)$.   
A typical goal is to compute the range $A^* = \arg\max_{A \in \c{A}} \Phi(A)$ and value $\Phi^* = \Phi(A^*)$ that maximizes the given function $\Phi$.  Our goal is to find a range $\hat A_\eps$ that satisfies 
$
\Phi(\hat A_\eps) \geq \Phi(A^*) - \eps.  
$

The exact version of this problem arises in many scenarios, formally as the classic discrepancy maximization problem~\cite{Backurs16,DE93}.  The rectangle version is a core subroutine in algorithms ranging from computer graphics~\cite{DEM96} to association rules in data mining~\cite{FMMT96}.   Also, for instance, in the world of discrepancy theory~\cite{Mat99,Cha01}, this is the task of evaluating how large the discrepancy for a given coloring is.  
For the halfspace setting, this maps to the minimum disagreement problem in machine learning (i.e., building a linear classifier)~\cite{Lin96}.
When $\Phi$ is replaced with a statistically motivated form~\cite{Kul97,Kul7.0}, then this task (typically focusing on disks or rectangles) is the core subroutine in the GIScience goal of computing the spatial scan statistic~\cite{HKG07,NM04,APV06,AMPVZ06} to identify spatial anomalies.  Indeed this statistical problem can be reduced the approximate variant with the simple discrepancy maximization form~\cite{APV06}.  

The approximate versions of these problems are often just as useful.  Low-discrepancy colorings~\cite{Mat99,Cha01} are often used to create the associated $\eps$-approximations of range spaces, so an approximate evaluation is typically as good.  It is common in machine learning to allow $\eps$ classification error.  In spatial scan statistics, the approximate versions are as statistically powerful as the exact version and significantly more scalable~\cite{SSSS}.  

While the exact versions take super-linear polynomial time in $m$, e.g., the rectangle version with linear functions takes $\Omega(m^2)$ time conditional on a result of Backurs~\etal~\cite{Backurs16}, we show approximation algorithms with $O(m + \poly(1/\eps))$ runtime.  
This improvement is imperative when considering massive spatial data, such as geotagged social media, road networks, wildlife sightings, or population/census data. In each case the size $m$ can reach into the 100s of millions. 

While most prior work has focused on improving the polynomials on the exact algorithms for various shapes \cite{Kulldorff2006, Tango2005} or on using heuristics to ignore regions~\cite{WSJRG09, NM04}, little work exists on approximate versions.  
These include 
\cite{AMPVZ06} which introduced generic sampling bounds, 
\cite{SSSS} which showed that a two-stage random sampling can provide some error guarantees, and 
\cite{walther2010} which showed approximation guarantees under the Bernoulli model.  
In this paper, we apply a variety of techniques from combinatorial geometry to produce significantly faster algorithms; see Table \ref{tbl:results}.  	

\subparagraph*{Our results.}
Our work involves constructing a two-part coreset of the initial range space $(X, \c{A})$; it approximates the ground set $X$ \emph{and} the set of ranges $\c{A}$.  This needs to be done in a way so that ranges can still be effectively enumerated and $\mu_R(A)$ and $\mu_B(A)$ values tabulated.  We develop fast coreset \emph{constructions}, and then extend and adapt exact scanning algorithms to the sparsified range space.  

We develop notation and review known solutions in Section \ref{sec:review}; also see Table \ref{tbl:results}.  Then we describe a general sampling result in Section \ref{sec:new-general} for ranges with bounded VC-dimension.  In particular, many of these results can be seen as formalizations and refinements (in theory and practice) of the two-stage random sampling ideas introduced in \cite{SSSS}.

In Section \ref{sec:halfspaces} we describe improvements for halfspaces and disks. We first improve upon the sampling analysis to approximate ranges $\c{H}_2$.  By carefully annotating and traversing the dual arrangement from the approximate range space, we improve further upon the general construction.

Then in Section \ref{sec:rect} we describe our improved results for rectangles.  We significantly extend the exact algorithm of Barbay \etal~\cite{Barbay2014} and obtain an algorithm that takes $O(m + \frac{1}{\eps^2} \log \frac{1}{\eps})$.  This is improved to $O(m + \frac{1}{\eps^2} \log \log \frac{1}{\eps})$ with some more careful analysis in Appendix \ref{app:loglog-linear}.
This nearly matches a new conditional lower bound of $\Omega(m + \frac{1}{\eps^2})$, assuming current algorithms for APSP are optimal~\cite{Backurs16}.

In Section \ref{sec:fxn-apx} we show how to approximate a \emph{statistical discrepancy function} (\textsc{sdf}, defined in Section \ref{sec:fxn-apx}) $\Phi$, as well as any \emph{general} function $\Phi$.  These require altered scanning approaches and the \textsc{sdf}-approximation requires a reduction to a number of calls to the generic (``linear'') $\Phi$.  We reduce the number of needed calls to generic $\Phi$ functions from $O(\frac{1}{\eps} \log \frac{1}{\eps})$~\cite{APV06} to $O(\frac{1}{\sqrt{\eps}})$.  

Finally, in Section \ref{sec:exp} we show on rectangles strong \emph{empirical} improvement over state of the art~\cite{SSSS}.

\begin{table}[b]
\vspace{-.2in}
\begin{tabular}{ c  c | c | c | c }
\hline
	& & Known Exact & Known Approx~\cite{SSSS} & \hspace{-.05in}\textbf{New Runtime Bounds} 
\\ \hline
	\hspace{-.05in}General Range Space& 
	& $O(m^{\nu + 1})$ 
	& --
	& $O\left (m + \frac{1}{\eps^{\nu + 2}} \log^{\nu} \frac{1}{\eps} \right)$ 

	\\

	Halfspaces & $\R^d$ 
	& $O(m^d)$~\cite{DEM96} 
	& --
	& $O \left (m + \frac{1}{\eps^{d + 1/3}} \log^{2/3} \frac{1}{\eps} \right)$ 

	\\ 
	Disks & $\R^2$ 
	& $O(m^{3})$~\cite{DEM96} 
	& $O(m + \frac{1}{\eps^4} \log^3 \frac{1}{\eps})$ 
	& $O \left ( m + \frac{1}{\eps^{3+1/3}}\log^{2/3} \frac{1}{\eps} \right )$ 

	\\ 
	Rectangles & $\R^2$ 
	& $O(m^2)$~\cite{Barbay2014} 
	& $O(m + \frac{1}{\eps^4} \log \frac{1}{\eps})$~\cite{APV06, AMPVZ06} 
	& $O(m + \frac{1}{\eps^2} \log\log \frac{1}{\eps})$

	\\ 
	Rectangles (\textsc{sdf}) & $\R^2$ 
	& $O(m^4)$ 
	& $O(m + \frac{1}{\eps^4} \log^4 \frac{1}{\eps})$ 
	& $O \left (m + \frac{1}{\eps ^{2.5}} \right )$  

	\\

	Rectangles (General) & $\R^2$ 
	& $O(m^4)$ 
	& $O(m + \frac{1}{\eps^4} \log^4 \frac{1}{\eps}))$ 
	& $O \left (m + \frac{1}{\eps^{4}} \right )$ 

	\\
	\hline
\end{tabular} 
\caption{\label{tbl:results} Algorithm times for ($\eps$-approximately) maximizing different range spaces.  Here dimension $d$, VC-dimension $\nu$, and probability of failure are all constants. For $(X,\c{R}_2)$ we show it takes $\Omega(m + 1/\eps^2)$ time, assuming hardness of APSP. }
  
\vspace{-.1in}

\end{table}

\section{Background on Geometric Range Spaces}
\label{sec:review}

\begin{figure}
	\includegraphics[width=\linewidth]{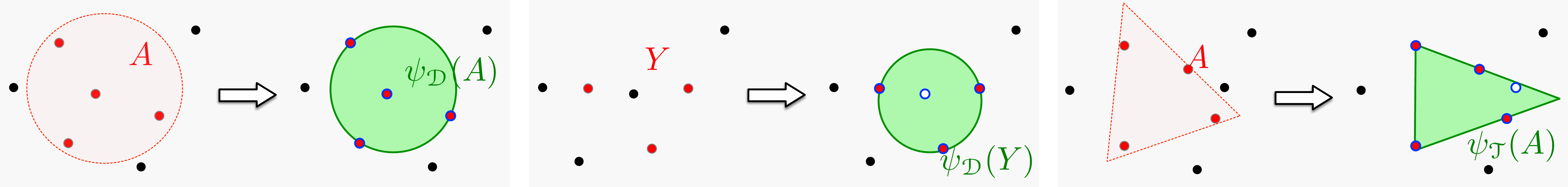}
	\vspace{-.3in}
	\caption{\label{fig:conforming} \s{First two panels show that $(\R^2, \c{D})$ has a conforming map $\psi_{\c{D}}$ defined by the smallest enclosing disk.  The last panel shows a range space $(X,\c{T})$ corresponding to triangles, and that a mapping $\psi_{\c{T}}$ defined by minimum area triangle is not conforming; it does not recover $A$.}}
	\vspace{-.2in}
\end{figure}

To review, a range space $(X,\c{A})$ is composed of a ground set $X$ (for instance a set of points in $\R^d$) and a family of subsets $\c{A}$ of that set.  
In this paper we are interested in geometrically defined range spaces $(X,
\c{A})$, where $X \subset \R^d$.  We formalize the requirements of this geometry via a conforming geometric mapping $\psi$; see Figure \ref{fig:conforming}.  
Specifically, it maps from a subset $Y \subset X$ to subset of $\R^d$. Typically, the result is a Lebesgue measureable subset of $\R^d$, for instance $\psi_{\c{D}}(Y)$, defined for disk range space $(X,\c{D})$, could map to the smallest enclosing disk of $Y$.  

We say this mapping $\psi_{\c{A}}$ is \emph{conforming to $\c{A}$} if for any $N \subset X$ it has the properties:
\begin{itemize} \denselist
\item 
    for any subset $A \in (N, \c{A})$ then $\psi_{\c{A}}(A) \cap N = A$  
    \hfill [\emph{the mapping recovers the same subset}]
\item
    for any subset $Y \subset X$ then $\psi_{\c{A}}(Y) \cap X \in (X, \c{A})$ 
    \hfill [\emph{the mapping is always in $(X,\c{A})$}]
\end{itemize}

\subsection{Basic Combinatorial Properties of Geometric Range Spaces}

We highlight two general combinatorial properties of geometric range spaces.  These are critical in sparsification of the data and ranges, and enumeration of the ranges.  

\subparagraph*{Sparsification.}
An \emph{$\eps$-sample} $S \subset X$ of a range space $(X,\c{A})$ preserves the density for all ranges as 
$
\max_{A \in \c{A}} |\frac{|X \cap A|}{|X|} - \frac{|S \cap A|}{|S|}| \leq \eps.  
$
An \emph{$\eps$-net} $N \subset X$ of a range space $(X,\c{A})$ hits large ranges, specifically for all ranges $A \in \c{A}$ such that $|X \cap A| \geq \eps |X|$ we guarantee that $N \cap A \neq \emptyset$.  
Consider range space $(X,\c{A})$ with VC-dimension $\nu$.   Then a random sample $S \subset X$ of size $O(\frac{1}{\eps^2} (\nu + \log \frac{1}{\delta})$ is an $\eps$-sample with probability at least $1-\delta$~\cite{VC71,LLS01}.  Also a random sample $N \subset X$ of size $O(\frac{\nu}{\eps} \log \frac{1}{\eps \delta})$ is an $\eps$-net with probability at least $1-\delta$.  
For our ranges of interest, the VC-dimensions of $(X,\c{H}_d)$, $(X, \c{D})$, and $(X, \c{R}_d)$ are $d$, $3$, and $2d$.

\subparagraph*{Enumeration.}
For the ranges spaces we will consider that each range can be defined by a \emph{basis} $B$; where $B$ is a point set. Given a geometric conforming map $\psi$ and subset $Y$, a range space's basis $B \subset Y$ is such that 
$\psi(B) = \psi(Y)$, but 
on a strict subset  $B' \subset B$, then $\psi(B')$ is different (and usually smaller under some measure) than $\psi(B)$.  
We will use $\beta$ to denote the maximum size of the basis for any subset $Y \subset X$.  For instance 
for $\psi_{\c{D}}$ then $\beta =3$, 
for $\psi_{\c{R}_d}$ then $\beta = 2d$, and
for $\psi_{\c{H}_d}$ then $\beta = d$.  
Recall, by Sauer's Lemma~\cite{Sau72}, if a range space $(X, \c{A})$ has VC-dimension $\nu$, then $\beta \leq \nu$.  

This implies that for $m=|X|$ points, there are at most ${m \choose \beta} = O(m^{\beta})$ different ranges to consider.  
We assume $\beta$ is constant; then it is possible to construct $\psi(Y)$ in $O(|Y|)$ time, and to determine if $\psi(Y)$ contains a point $x \in X$ in $O(1)$ time. 
This means we can enumerate all $O(m^\beta)$ possible bases in $O(m^{\beta})$ time, construct their maps $\psi(B)$ in as much time, and for all of them count which points are inside, and evaluate each $\Phi(A)$ to find $A^*$, in $O(m^{\beta+1})$ time.  

For the specific range spaces we study, the time to find $A^* \in \c{A}$ can be improved by faster enumeration techniques.  
For $\c{H}_d$, Dobkin and Eppstein~\cite{DE93} reduced the runtime to find $A^*$ from $O(m^{d+1})$ to $O(m^d)$; this implies for $\c{D}$ the runtime is reduced from $O(m^4)$ to $O(m^3)$.  
For $\c{R}_d$, Barbay \etal~\cite{Barbay2014} show how to find $A^*$ in $O(m^d)$ time; this was recently shown tight~\cite{Backurs16} in $\R^2$, assuming APSP takes cubic time.

\subsection{Coverings}

Our main approach towards efficient approximate range maximization, is to sparsify the range space $(X,\c{A})$.  This will have two parts.  The first is simply replacing $X$ with an $\eps$-sample.  The second is sparsifying the ranges $\c{A}$, using a concept we refer to as an $\eps$-covering.

Recall that the symmetric difference of two sets $A \triangle B$ is $(A \cup B) \setminus (A \cap B)$.  
Define an \emph{$\eps$-covering} $(X, \c{A}_\triangle)$ of a range space $(X, \c{A})$ where  $(X,\c{A}_\triangle) \subset (X,\c{A})$, so  that for any $A \in \c{A}$ there exists a $A' \in \c{A}_\triangle$ such that 
$
|A \triangle A'| \leq \eps |X|.
$ 
See Figure \ref{fig:covering} for an illustration of this concept.  If a range space satisfies the above condition for any one specific range $A$, but not necessarily all ranges $A \in \c{A}$ simultaneously, then it is a \emph{weak $\eps$-covering} of $(X,\c{A})$.  

We will use subsets of the ground set to define subsets of the ranges. For a subset $N \subset X$, let $\c{A}_{\mid N} = \{A \cap N \mid A \in \c{A}\}$ be the restriction of $\c{A}$ to the points in $N$.  We will define $(X, \c{A}_\triangle)$ using $\c{A}_{\mid N}$ or a subset thereof.  However, as each $A \in \c{A}_{\mid N}$ is a subset of $N$, which itself is a subset of $X$, we need a conforming map $\psi_{\c{A}}$ to take a region $A \in \c{A}_\triangle$ and map it back to some region in $\c{A}$, a subset of $X$.  Given $\c{A}'_{\mid N}$ (which is $\c{A}_{\mid N}$ or a subset) we define $(X, \c{A}_{\triangle})$ as
\[
(X, \c{A}_\triangle) = \{ X \cap \psi_{\c{A}}(A) \mid A \in (N, \c{A}'_{\mid N}) \}.
\]

\begin{figure}
\includegraphics[width=\linewidth]{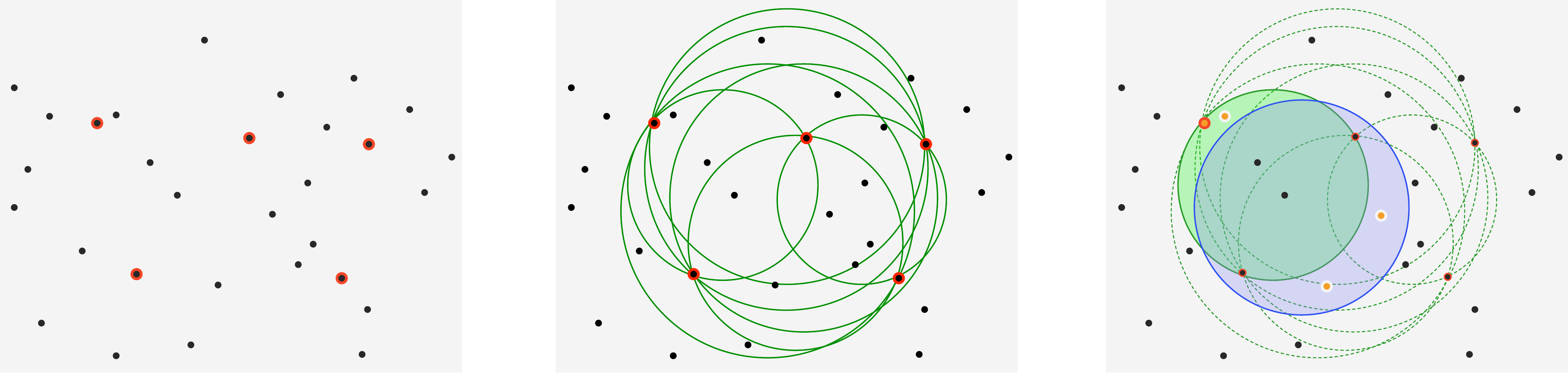}
	\vspace{-.3in}
	\caption{\label{fig:covering} \s{First panel shows $N \subset X$.  Second panel shows the set of disks $\{\psi_{\c{D}}(A) \mid A \in (N,\c{D}_{\mid N})\}$ induced by $N$.  The third panel shows a range $Y \subset X$ (defined by disk in blue).  It has symmetric difference over $X$ (in orange) of size $4$ with the one defined by the disk $\psi_{\c{D}}(A)$ (in green)  induced by a subset $A \subset (N,\c{D}_{\mid N})$.}}
	\vspace{-.2in}
\end{figure}

A small sized $\eps$-covering is implied by a result of Haussler~\cite{Hau95}. For every range space $(X, \c{A})$ of VC-dimension $\nu$, with $m=|X|$, there always exist a maximal set of ranges $A_\triangle$ of size $O((\frac{m}{k + \nu})^{\nu})$ where for every pair of ranges $A, A'\in A_\triangle$ the symmetric difference $|A \triangle A'| \ge k$. Setting $k = m \eps$ then $(\frac{m}{k + \nu})^{\nu} = O(\frac{1}{\eps^\nu})$, so $A_{\triangle}$ is an $\eps$-covering.

\subparagraph*{Symmetric difference nets.}  
We can construct an $\eps$-net over the symmetric difference range space of $\c{A}$ and then use these points to define $\c{A}_\triangle$.  
 
For a family of ranges $\c{A}$, let $\c{S}_{\c{A}}$ be the family of ranges made up of the symmetric difference of ranges of $\c{A}$. Specifically $\c{S}_{\c{A}} = \{A_1 \triangle A_2 \mid A_1, A_2 \in \c{A}\}$.  If range space $(X, \c{A})$ has VC-dimension $\nu$, then $(X, \c{S}_{\c{A}})$ has VC-dimension at most $O(\nu \log \nu)$~\cite{Mat02}.  Thus for constant $\nu$ we can use asymptotically the same size random sample as before.  Matheny \etal~\cite{SSSS} pointed out two important properties connecting nets over symmetric difference range spaces and $\eps$-coverings and then finding $\hat A_\eps$.  
\begin{itemize} \denselist
\item[(P1)]
An $\eps$-net $N$ for $(X, \c{S}_{\c{A}})$ induces $(N, \c{A}_{\mid N})$ which is an $\eps$-covering of $(X, \c{A})$~\cite{SSSS}.

\item[(P2)]
	Given an $\frac{\eps}{2}$-covering $(N, \c{A}_\triangle)$ and an $\frac{\eps}{2}$-sample $S$ over $(X, \c{A})$ then for any range $A \in (X,\c{A})$, there exists a range $\psi_{\c{A}}(A') \cap X$ for $A' \in \c{A}_{\mid N}$ so 
$
	\left | \frac{|A \cap X|}{|X|} - \frac{|\psi_{\c{A}}(A') \cap S|}{|S|} \right | \leq \eps  
$~\cite{SSSS}.
\end{itemize}

For an appropriate constant $C$, by constructing $(\eps/C)$-nets $N_R$ and $N_B$, of size $n$, on the red $(R,\c{S}_{\c{A}})$ and blue $(B,\c{S}_{\c{A}})$ points, also constructing $(\eps/C)$-samples of size $s$ on $(R,\c{A})$ and $(B,\c{A})$, and invoking (P2) on the results, Matheny \etal~\cite{SSSS} observed we can maximize $\Phi(\psi_\c{A}(A') \cap S)$ over $A' \in \c{A}_{\mid N_R} \cup  \c{A}_{\mid N_B}$ to find an $\eps$-approximate $\hat A_\eps$.  
They construct the $\eps$-nets and $\eps$-samples using random sampling, and apply the results to scan disk $\c{D}$ and rectangle $\c{R}_2$ range spaces towards finding $\hat A_\eps$.  Enumerating all ranges in $A' \in \c{A}_{\mid N_R} \cup \c{A}_{\mid N_B}$ and counting the intersections with the $(\eps/C)$-samples, when $C$ is a constant,  is sufficient to find an $\hat A_\eps$ in time 
$
O(m + |N|^2|S| \log n) = O(m + \frac{1}{\eps^4}\log^3 \frac{1}{\eps})
$ 
for disks $(X, \c{D})$ and time
$
O(m + |N|^4 + |S| \log n) = O(m + \frac{1}{\eps^4}\log^4 \frac{1}{\eps}) 
$
for rectangles $(X,\c{R}_2)$.  

We can ignore the distinct red and blue points, and focus on three aspects of this problem which can be further optimized:
(1) More efficiently constructing a sparse set of $\eps$-covering ranges $(X, \c{A}_\triangle)$.
(2) More efficiently constructing a smaller $\eps$-sample $S$ of $(X, \c{A})$.  
(3) More efficiently scanning the resulting $(S, \c{A}_\triangle)$.  

\section{General Results via $\eps$-Coverings}
\label{sec:new-general}

For general range spaces of contant VC-dimension $\nu$ we can directly apply the work of Matheny \etal~\cite{SSSS} to get a bound.  
A random sample $N$ of size $O(\frac{\nu \log \nu}{\eps} \log \frac{\nu}{\eps})$ induces an $\eps$-covering $(X,\c{A}_{\mid N})$ with constant probability by (P1).  
A random sample $S$ of size $O(\frac{\nu}{\eps^2})$ induces an $\eps$-sample with constant probability. By (P2), scanning the ranges in $(X, \c{A}_{\mid N})$, evaluating $\Phi(A)$ on each ranges $A$ using $S$, and returning the maximum $\hat A_\eps$ induces the $\eps$-approximation of $\Phi(A^*)$ as we desire.  Including the time to calculate $N$ and $S$ we obtain the following result.

\begin{theorem}
\label{theorem:general-range}
Consider a range space $(X,\c{A})$ with constant VC-dimension $\nu$, with $|X| =m$, and conforming map $\psi_\c{A}$.  For 
$A^* = \arg\max_{A \in \c{A}} \Phi(A)$, with probability at least $1-\delta$, in time $O(m + \frac{1}{\eps^{\nu+2}} \log^\nu \frac{1}{\eps} \log \frac{1}{\delta})$, we can find a range $\hat A_\eps$ so that 
$
|\Phi(A^*) - \Phi(\hat A_\eps)| \leq \eps.
$
\end{theorem}
\begin{proof}
First compute random samples $N$ and $S$ of size $O(\frac{1}{\eps} \log \frac{1}{\eps})$ and $O(\frac{1}{\eps^2})$ respectively.  
The algorithm naively considers all $O((\frac{1}{\eps} \log \frac{1}{\eps})^\nu)$ subsets $B \subset N$ of size $\nu$, and calculates the quantity $\Phi(S \cap \psi_{\c{A}}(B))$.  By (P2), this can be used to $\eps$-approximate $\Phi(A)$ for any range $A \in \c{A}$ which has less than $\eps$-symmetric difference with $\psi_{\c{A}}(B)$.  Moreover, since $(X,\c{A}_{\mid N})$ is an $\eps$-cover, with constant probability any range $A$ is within symmetric difference of at most $\eps m$ of one induced by some subset $B$.  Thus, with constant probability we observe some range $\hat A_\eps = X \cap \psi_{\c{A}}(B)$ for which $|\Phi(A^*) - \Phi(\hat A_\eps)| \leq \eps$ (after adjusting constants in the size of $N$ and $S$).  
To amplify the probability of success to $1-\delta$, we repeat this process $O(\log \frac{1}{\delta})$ times, and return the $\hat A_\eps$ with median score.    
\end{proof}

\section{Halfspaces}
\label{sec:halfspaces}

Our general additive error results applied to arbitrary halfspaces, $(X, \c{H}_d)$, would require $O(m + \frac{1}{\eps^{d+2}} \log^{d} \frac{1}{\eps} \log \frac{1}{\delta})$ time.  
In this section, we improve this runtime to $O(m + \frac{1}{\eps^{d + 1}} \log \frac{1}{\delta})$. 
First, a recent paper~\cite{MP18} shows that with constant probability an $\eps$-sample $S$ for $(X,\c{H}_2)$ of size $s = O(\frac{1}{\eps^{4/3}} \log^{2/3} \frac{1}{\eps})$ can be constructed in $O(m + \frac{1}{\eps^2} \log(\frac{1}{\eps}))$ time.

Second we create a weak $\eps$-covering of $(X, \c{A})$ using $(X, \c{A}_{\mid N})$ where $N$ is of size $O(1/\eps)$. Ultimately this requires time $O(m + \frac{1}{\eps^{d + 1/3}} \log^{2/3} \frac{1}{\eps})$, with constant probability.  

Then, we show how to enumerate these ranges while maintaining the counts from $S$ with less overhead than the previous brute force approaches.

\subsection{Smaller Coverings}

We show that a random sample $N$ of only $O(1/\eps)$ points induces a range space $(X, \c{A}_{\mid N})$ which is a weak $\eps$-covering of $(X, \c{H}_d)$.  

\begin{lemma}
	\label{lemma:small-halfspace-covering}
	A random sample $N$ of size $|N| = O(\frac{d^2}{\eps} \log d)$ induces a weak $\eps$-covering $(X,{\c{H}_d}_{\mid N})$ for $(X,\c{H}_d)$ with constant probability.
\end{lemma}
\begin{proof}
	
	For any halfspace range $H \in \c{H}_d$, we aim to show that $(X, {\c{H}_d}_{\mid N})$ contains some $H'$ within $\eps m$ symmetric difference of $H$ (recall $m = |X|$).  
	For any halfspace range $H$, we can translate and rotate this to define the geometric halfspace $h$, until it has at least $d$ points $\{x_1, x_2, \ldots, x_d\}$ incident to its boundary (some may be inside $H$ and some may be just outside).  We will let $h$ denote the geometric shape and $H$ the range defined $H = X \cap h$.  
	
	We can fix the last $d-1$ points $B_1 = \{x_2, \ldots, x_d\}$, and consider a rotation of $h$ so that those $d-1$ points stay incident to its boundary.  This defines an ordering over all points in $X \setminus B_1$.  Denote the  first $2\eps m/d$ points closest to $x_1$ in this ordering as the set $Y_1$. If we take the union of any $y \in Y_1$ with $B_1$ it induces another range $H_1'$ that has symmetric difference of size at most $\eps m/d$ with $H$.  
	A randomly chosen point from $X$ is in this set $Y_1$ with probability $2\eps/d$.  
	If we randomly choose $k$ points iid, then none of these will be from $Y_1$ with probability at least $\delta' = (1-2\eps/d)^k \leq \exp(-2k \eps/d)$.  Hence setting $k = \frac{d}{2\eps} \log \frac{1}{\delta'}$ will result in one point within $Y_1$ with probability at least $1-\delta'$.  
	Let the halfspace induced by $B_1 \cup y_1$ be $\hat h_1$.  
	
	If we assume we have chosen some point $y_1 \in Y_1$ from this first set of $k$ points, then we 
	fix this point in $B_2$, and also put the last $d-2$ points from $B_1$ in $B_2$.  This again results in $d-1$ points in $B_2$, and we can order the points to rotate a halfspace $h_2$ that is incident to these points.  We put the points within $\eps/d$ of $x_2$ in this ordering in $Y_2$.  Then again if we iid sample another $k$ points from $X$, with probability at least $1-\delta'$ one falls into $Y_2$.  Let such a point be $y_2$ and the halfspace induced by $y_2 \cup B_2$ is $\hat h_2$ and the symmetric difference $|(X \cap \hat h_1) \triangle (X \cap \hat h_2)| \leq \eps m/d$.  
	
	We repeat this for $d-2$ more steps, until we obtain a set $\{y_1, y_2, \ldots, y_d\}$.  This succeeds with probability at least $1-d \cdot \delta'$.  The induced halfspace space $\hat h_d$ is within symmetric difference $|(X \cap \hat h_d) \triangle (X \cap \hat h)| \leq d \cdot \eps m /d = \eps m$, by triangle inequality.  
	
	Setting $\delta' = \delta/d$, and $k = \frac{d}{\eps} \log \frac{d}{\delta}$ this implies after $\frac{d^2}{\eps} \log \frac{d}{\delta}$ iid samples from $X$, with probability at least $1-\delta$, the resulting set $N$ induces a weak $\eps$-covering of $(X, \c{H}_d)$.  
\end{proof}

\subsection{Fast Enumeration of Halfspaces}
\label{sec:H-enum}

Now using our sets $|N| = n$ and $|S|=s$ we enumerate over the ranges in the weak $\eps$-covering $O(X, {\c{H}_d}_{\mid N})$, and for each range we count the intersection with an $\eps$-sample $S$ of $(X, \c{H}_d)$.  We first consider the case when $d=2$; the general case will reduce to this case.  

Our technique follows that of Dobkin and Eppstein~\cite{DE93}.  This first builds the dual arrangement $A(N^*)$, where in the dual the points in $N$ are halfspaces in $N^*$.  Each halfspace $h \in N^*$ intersects at most $O(n)$ other halfspaces, and takes as long to insert in the arrangement; thus construction of $A(N^*)$ takes total $O(n^2)$ time.  Also, $A(N^*)$ has $O(n^2)$ vertices and edges, each edge representing a combinatorial range from ${\c{H}_2}_{\mid N}$.  At a vertex, we are incident to $4$ edges, these correspond to $4$ ranges in ${\c{H}_2}_{\mid N}$ that only toggle the inclusion of the two relevant points $x_1, x_2 \in N$ whose dual halfspaces $h_1$ and $h_2$ cross at that vertex.  

Then our technique extends that of Dobkin and Eppstein~\cite{DE93} in that we annotate each edge of $A(N^*)$ with each halfspace $g \in S^*$, the dual set of $S$.  Each such $g \in S^*$ intersects at most $O(n)$ edges of $A(N^*)$ and these edges can be found and annotated in $O(n)$ time.  Annotating all takes $O(ns)$ time.  
This annotation describes how many halfspaces are crossed when moving between vertices in the arrangement.  By considering the counts of $S$ at all vertices of $A(N^*)$, we evaluate $S$ on all ranges in $(S, {\c{H}_2}_{\mid N})$.  

We can traverse $A(N^*)$ using a topological sweep~\cite{EDEL89} in $O(n^2)$ time.  Starting at the far left, corresponding to ranges that contain no points, we maintain at each vertex of $A(N^*)$ how many halfspaces $g \in S^*$ we are below, and thus how large is $S \cap H$ for all $H \in (S, {\c{H}_2}_{\mid N})$.

Combining the results for $n = O(1/\eps)$ from Lemma \ref{lemma:small-halfspace-covering} 
and 
$s = O(1/\eps^{\frac{4}{3}} \log^{\frac{2}{3}} 1/\eps)$~\cite{MP18} both for constant $d$, 
we can state our result for $\c{H}_2$.  

\begin{lemma}
	Consider a range space $(X,\c{H}_2)$ with $|X| =m$.  For 
	maximum range $H^* = \arg\max_{H \in \c{H}_2} \Phi(H)$, with constant probability, in time 
	$O(m + \frac{1}{\eps^{7/3}}\log^{2/3} \frac{1}{\eps})$, 
	we can find a range $\hat H_\eps$ so that 
	$
	|\Phi(H^*) - \Phi(\hat H_\eps)| \leq \eps.
	$
\end{lemma}

To extend this to $(X, \c{H}_d)$ we start with a weak $\eps$-covering $(X, {\c{H}_d}_{\mid N})$ of size $n = O(\frac{1}{\eps})$ and a random sample $S$ of size $O(\frac{1}{\eps^2})$ which is an $\eps$-sample with constant probability~\cite{LLS01}.

To search $(X, {\c{H}_d}_{\mid N})$ we consider all $O(n^{d-2})$ subsets $L \subset N$ of size $d-2$.  For each subset $L$ we define a projection $\pi_L$ orthogonal to the span of $L$, resulting in a $2$-dimensional space.  Restricting to this $2$-dimensional space, we follow Dobkin and Eppstein~\cite{DE93}, and scan ${\c{H}_d}_{\mid N}$ under this projection, which enforces they have boundary incident to $L$.  Here we can again construct each pertinent dual arrangement in $2$ dimensions in $O(n^2)$ time.  

We can then reduce $S$ to a set $S_L$ of size $s = O(\frac{1}{\eps^{4/3}}\log^{2/3} \frac{1}{\eps})$, which is an $\eps$-sample restricted to $L$.  This takes $O(\frac{1}{\eps^2} \log \frac{1}{\eps})$ time for each $L$~\cite{MP18}.  

Then we annotate the $2$-dimensional dual arrangement of $\pi_L(N)$ with $S_L$ in $O(ns) = O(\frac{1}{\eps} \cdot \frac{1}{\eps^{4/3}}\log^{2/3} \frac{1}{\eps})$ 

time.  The topological sweep, and evaluation of each range takes $O(n^2)$ time, and does not dominate the cost.  
Ultimately for all $O(n^{d-2})$ subsets $L$, the total running time is $O(n^{d-2} (n s + \frac{1}{\eps^2} \log \frac{1}{\eps})) = O(\frac{1}{\eps^{d + 1/3 }}\log^{2/3} \frac{1}{\eps})$, for constant $d$.  
To amplify the success probability to at least $1-\delta$, we repeat the entire construction $O(\log \frac{1}{\delta})$ times, and return the $\hat H_\eps$ with median value $\Phi(\hat H_\eps)$.

\begin{theorem}
	\label{thm:halfspace}
	Consider a range space $(X,\c{H}_d)$ with $|X| =m$.  For 
	$H^* = \arg\max_{H \in \c{H}_d} \Phi(H)$, with probability at least $1-\delta$, in time $O(m + \frac{1}{\eps^{d + 1/3}} \log^{2/3} \frac{1}{\eps} \log \frac{1}{\delta})$, we find a range $\hat H_\eps$ so 
	$
	|\Phi(H^*) - \Phi(\hat H_\eps)| \leq \eps.
	$
\end{theorem}

\subsection{Application to Disks and other Ranges}

Many other geometric ranges can be mapped to halfspaces through various lifting maps including disks, ellipses, slabs, and annuli.  
For disks, we can solve the approximate maximum range problem for $(X, \c{D}_2)$ by using the lifting $(x, y) \rightarrow (x, y, x^2 + y^2)$.  Then invoking the result for $(X, \c{H}_3)$ from Theorem \ref{thm:halfspace} takes $O(m + \frac{1}{\eps^{3+1/3}}\log^{2/3}\frac{1}{\eps})$ time.

\section{Rectangles}
\label{sec:rect}

For the case of rectangles $(X, \c{R}_d)$, we will describe two classes of algorithms.  One simply creates an $\eps$-cover $(X, {\c{R}_d}_{\mid N})$ and evaluates each rectangle $A$ in this cover on an $\eps$-sample $S$ as before.  
The other takes specific advantage of the orthogonal structure of the rectangles and of ``linearity'' of $\Phi$; this algorithm can find the maximum in $\Phi$ among ranges in $(X, {\c{R}_d}_{\mid N})$ without considering every possible range.   Our techniques are inspired by several algorithms~\cite{Barbay2014,TAKAOKA2002,DEM96} for the exact maximization problem, but requires new ideas to efficiently take advantage of using both $N$ and $S$.  
Common to all techniques will be an efficient way to compute an $\eps$-cover based on a grid.


\subsection{Grid $\eps$-Covers for Rectangles}
\label{sec:grid}

We will create a grid $G$ via a set of $r = O(1/\eps)$ cells along each axis.  The grid is then the cross-product of these cells on each axis.  
The grid is defined so no row or column contains more than $\eps m / (2d)$ points. Given any rectangle $A$ define $A' \subset \R^d$ to be $A$ rounded to these grid boundaries.  Rounding $A$ to $A'$ incurs at most $\eps m / (2d)$ change in points on each side, so that in total for all $2d$ sides there is  $|(X \cap A) \triangle (X \cap A')| \leq 2d \cdot \frac{\eps m}{2d} = \eps m$ error.  

We can sort $X$ along each axis in $O(m \log m)$ time, and take a set $N_i$ of $r-1$ points along the $i$th axis of points evenly spaced in this sorted order.  These define the grid boundaries on each axis.  If we choose $r-1 = m \cdot 4d/\eps$, then no row contains more than $\eps m / (2d)$ points as desired.  

We label the rectangular ranges of $X$ restricted to this grid boundary as $(X, {\c{R}_d}_{\mid G})$, and as argued above it is an $\eps$-cover of $(X, \c{R}_d)$.  
We can then efficiently annotate each grid cell with approximately how many points it contains from $X$.  
For each point $x \in X$ we assign it to the count of each grid cell in $O(\log \frac{1}{\eps})$ time, for constant $d$.  Since any rectangle in $(X, {\c{R}_d}_{\mid G})$ is also a rectangle in $(X, \c{R}_d)$ then the count on any rectangle in $(X, {\c{R}_d}_{\mid G})$ can be estimated within $\eps m$ by examining $G$.  
\begin{lemma}
	\label{lem:grid-cover}
	For range space $(X, \c{R}_d)$ where $|X|=m$, the construction of grid $G$ takes $O(m \log m + \frac{1}{\eps^d})$ time, has $O(1/\eps)$ cells on each side, and induces an $\eps$-cover $(X, {\c{R}_d}_{\mid G})$ of $(X, \c{R}_d)$ for constant $d > 1$.  
\end{lemma}

\subparagraph*{Simple enumeration algorithm.}
Given the point set $m$ we take an $(\eps /4d)$-sample $S$ and then construct a grid $G$ on it
in $O(m + \frac{1}{\eps^2} \log \frac{1}{\eps} + \frac{1}{\eps^d})$ time. $S$ is also an $(\eps/4d)$-sample for intervals along each axis, so running our grid construction on $S$ induces an $\eps$-cover 
of $(X, \c{R}_d)$. 
Once we have the approximate count in each grid cell of $G$, we can build subset sum information.  That is, in each of $O(1/\eps^d)$ grid cells, we compute the sum of counts for all grid cells with smaller or equal indexes in each dimension.  Straightforward dynamic programming solves this in $O(1/\eps^d)$ time.  
Then for any rectangle on $(S, {\c{R}_d}_{\mid G})$ can have its count calculated in $O(1)$ time using inclusion-exclusion formulas from a constant number of subset sum values; for instance when $d=2$, $4$ values are required.  

We evaluate all rectangles by enumerating all pairs of grid cells (defining upper and lower corners in each dimension), and calculating the count for $S$ (technically a separate red count from $S_R$ and blue count from $S_B$) in $O(1)$ time.  

\begin{theorem}
	\label{thm:rect-gen}
	Consider a range space $(X, \c{R}_d)$ with $|X|=m$ and an Lipschitz-continuous function $\Phi$ with maximum range $A^* = \arg \max_{A \in \c{R}_d} \Phi(A)$.  
	With probability at least $1-1/e^{1/\eps}$, in time $O(m + \frac{1}{\eps^{2d}})$ we can find a range $\hat A_\eps$ so that 
	$
	|\Phi(A^*) - \Phi(\hat A_\eps)| \leq \eps.
	$
\end{theorem}

\subsection{Algorithms for Decomposable Functions}
\label{subsection:linear-algorithm}

Here we exploit a critical ``linear'' property of $\Phi$ that 
a rectangle $A$ can be decomposed into any two parts $A_1$ and $A_2$ and 
$\Phi(A) = \Phi(A_1) + \Phi(A_2)$.  Technically, we solve both $\Phi^+(A) = \mu_R(A) - \mu_B(A)$ and $\Phi^-(A) = \mu_B(A) - \mu_R(A)$ separately, and take their max.  
In particular, this allows us (following exact algorithms~\cite{Barbay2014}) to decompose the problem along a separating line.  The solution then either lies completely on one half, or spans the line.  In the exact case on $s$ points, this ultimately leads to a run time recurrence of $\c{T}_1(s) = 2\c{T}_1(s/2) + \c{T}_2(s)$ where $\c{T}_2(s)$ is the time to compute the problem 
spanning the line. The line spanning problem can then be handled using a different recurrence that leads to $\c{T}_2(s) = O(s^2)$ and a total runtime for the problem
of $\c{T}_1(s) = 2\c{T}_1(s/2) + O(s^2) = O(s^2)$ \cite{Barbay2014}.  

First we show we can efficiently construct a special sample $S$ of size $s = O(1 /(\eps^2 \log \frac{1}{\eps}))$, but this still would requires runtime of roughly $1/\eps^4$.  

Our approximate algorithm will significantly improve upon this be compressing the representation at various points, but requiring some extra bookkeeping and a bit more complicated recurrence to analyze.  
In short, we can map $S$ to an $r \times r$ grid (using Lemma \ref{lem:grid-cover}), and then the recurrence only depends on the dyadic y-intervals of the grid.  We can compress each such interval to have only $\eps s / \log r$ error, since each query only touches about $\log r$ of these intervals.  The challenge then falls to maintaining this compressed structure more efficiently during the recurrence.

The dense exact case on an $r \times r$ grid is also well studied.  There exists a practically efficient $O(r^3)$ time method~\cite{Bentley84} based on Kadane's algorithm (which performs best as \gridScanLin; see Section \ref{sec:exp}), and a more complicated method taking $O(r^3 ( \frac{\log \log r}{\log r} )^{\frac{1}{2}})$ time~\cite{TAKAOKA2002}.  
By allowing an approximation, we ultimately reduce this runtime to $O(r^2 \log r) = O(\frac{1}{\eps^2} \log \frac{1}{\eps})$.  

We will focus on the 2d case. This is where the advantage over the Theorem \ref{thm:rect-gen} bound of $O(m + 1/\eps^4)$ is most notable.  Generalization to high dimensions is straightforward: enumerate over pairs of grid cells to define the first $d-2$ dimensions, then apply the $2$-dimensional result on the remaining dimensions.

\subparagraph*{Tree and slab approximation.}
The algorithm builds a binary tree over the rows (the $y$ values) of $G$.  
We will assume that the number of cells in each axis $r = O(1/\eps)$ is a power of $2$ (otherwise we can round up), so it is a perfectly balanced binary tree.  

At the $i$th level of the tree, each node contains $r/2^i$ rows and there are $2^i$ nodes.  We refer to the family of rows represented by a subtree as a \emph{slab}.  
Any grid-aligned rectangle $A = [x_1, x_2] \times [y_1, y_2]$ can be defined as the intersection of $[x_1, x_2]$ with at most $2 \log_2 r$ slabs in the y-coordinate -- the classic dyadic decomposition.  This implies we can tolerate $\eta s = O(\eps s/\log r)$ additive error in each slab to have at most $O(\eps s)$ additive error overall (which implies the percentage of red and of blue points in each range has additive $O(\eps)$ error).

\begin{wrapfigure}{R}{0.25\textwidth}
	\vspace{-.15in}
	\includegraphics[width=0.3\textwidth]{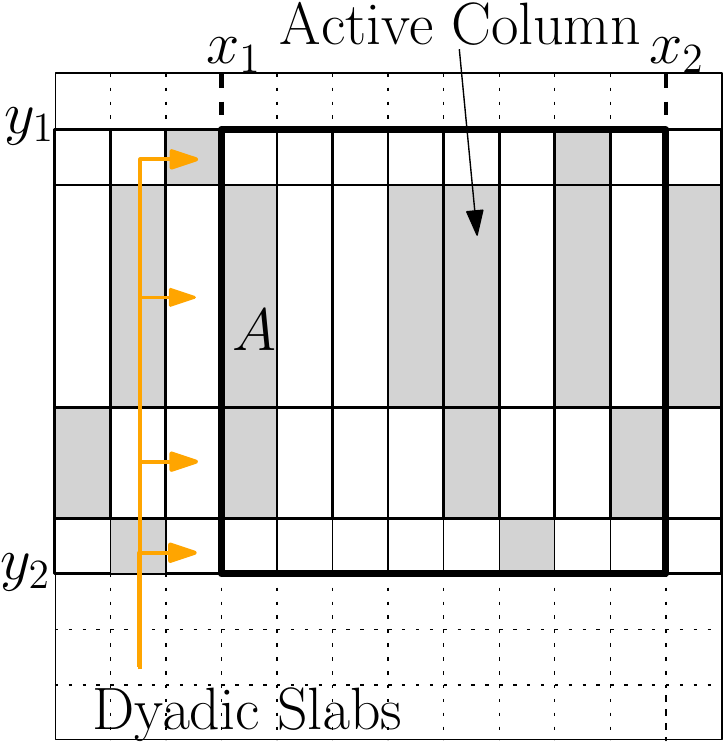}
	\vspace{-.35in}
	\label{fig:TreeSlab}
\end{wrapfigure}

Since the rectangle will span the entire vertical extent ($y$ direction) of each slab in this decomposition, the additive error of a slab can be obtained along just the horizontal ($x$) direction.  Thus, we can scan cells from left to right within a slab, and only retain the cumulative weight in a cell when it exceeds $\eta s$. We refer to this operation as \textit{$\eta$-compression}.  We denote each column (and $x$ value) within a slab where it has retained a non-zero value as \emph{active}, all other columns are \emph{inactive}.  We store the active cells in a linked list.

Since there are $\Theta(s/r)$ points per row, it implies we can approximate each slab consisting of $1$ row (a leaf of the tree, level $\log_2 r$) with weights in only $O(1/(r \eta)) = O( \log r)$ cells (since $r = O(\frac{1}{\eps})$).  And a slab at level $i$ (originally with $\Theta(s/2^i)$ points) can be approximated by accumulating weight in $O(\min\{r,1 / (\eta 2^i)\})$ cells.  For level $i > \log 1/ \eta r$, this compresses the points in that slab.

\begin{lemma} \label{lem:compression}
In $O(r^2)$ time, we can compress all slabs in the tree, so a slab at level $i$ contains $\ell_i = O(\min\{r, 1/ (\eta 2^i) \})$ active columns where
$\eta = O(\eps / \log r)$.  
\end{lemma}
\subparagraph*{Interval Preprocessing and Merging.}
Now consider a subproblem, where we seek to find a rectangle $A = [x_1, x_2] \times [y_1, y_2]$ to maximize the total weight, restricted to a given horizontal extent $[y_1, y_2]$ (e.g., within a slab).  We reduce this to a $1$d problem by summing the weights for each $x$-coordinate to $w_x = \sum_{y \in [y_1, y_2]} w_{x,y}$.  
Then there is an often-used~\cite{Barbay2014,DE93,APV06} way to preprocess intervals $[x_1', x_2']$ so they can be merged and updated.  It maintains $3$ maximal weight subintervals:  
(1) the maximal weight subinterval in $[x_1', x_2']$, 
(2) the maximal weight interval including the left boundary $x_1'$, and
(3) the maximal weight interval including the right boundary $x_2'$.  
Given two preprocessed adjacent intervals $[x_1', x_2']$ and $[x_2'+1, x_3']$, we can update these subintervals to $[x_1', x_3']$ in $O(1)$ time~\cite{Barbay2014}.  
Thus given a horizontal extent with $a$ active intervals, we can find the maximum weight subinterval in $O(a)$ time.  

\subparagraph*{Recursive construction.}
Now we can describe our recursive algorithm for finding the maximal weight rectangle on the grid $G$.  We find the maximum weight rectangle 
through 3 options:
(1) completely in the top child's subtree, 
(2) completely in the bottom child's subtree,
(3) overlapping both the top and bottom child's subtree.  
The total time can be written as a recurrence as $\c{T}_1(r) = 2 \c{T}_1(r/2) + \c{T}_2(r)$, where $\c{T}_2$ is the time to solve case (3).  

Case (3) requires another recurrence to understand, and it closely follows the ``strip-constrained'' algorithm of Barbay \etal~\cite{Barbay2014}; our version will account for the dense grid.  

We consider the \textsc{Strip-constrained grid search} problem: 
\emph{First fix a \emph{strip} $M$ which is a consecutive set of rows.  Then consider two slabs $T$ and $B$ where $T$ is directly above (on \textbf{t}op of) $M$ and $B$ is directly \textbf{b}elow $M$.  A column of $M$ is active if it is active in $T$ or $B$.  Counts in active columns of $M$ are maintained, and intervals of $M$ described by consecutive inactive columns have been merged.  
The goal is to find the maximum weight rectangle with vertical span $[y_1, y_2]$ where $y_2$ is in $T$ and $y_1$ is in $B$ (it must cross $M$).  }

\begin{wrapfigure}{r}{.8in}
\vspace{-.35in}
\includegraphics[width=.8in]{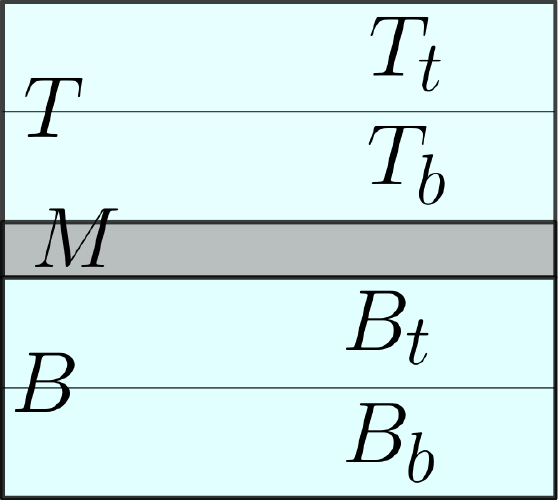}
\vspace{-.5in}
\end{wrapfigure}
We specifically want to solve this problem when $M$ is empty, $T$ is the top child and $B$ the bottom child of the root, and all columns are initially active.  We call this the case of size $r$ since there are still $r$ rows.  

\begin{lemma}
\label{lem:rect-strip}
The Strip-constrained grid search problem of size $r$ over an $\eta$-compressed binary tree takes $O(r / \eta)$ time.  
\end{lemma}

\begin{proof}
Following Barbay \etal~\cite{Barbay2014} we split the problem into $4$ subcases, following the subtrees of the slabs.  Slab $T$ has a top $T_t$ and bottom $T_b$ sub-slab, and similarly $B_t$ and $B_b$ for $B$.  Then we consider $4$ recursive cases with new strip $M'$:
(1) slabs $T_t$ and $B_b$ with $M' = T_b \cup M \cup B_t$, 
(2) slabs $T_b$ and $B_b$ with $M' = M \cup B_t$, 
(3) slabs $T_t$ and $B_t$ with $M' = T_b \cup M$, and
(4) slabs $T_b$ and $B_t$ with $M' = M$.  
The cost in a recursive step is the preprocessing of the new slab $M'$.  We will describe the largest case (1); the others are similar.

Strip $M$ already maintains preprocessed intervals of inactive columns.  When $T_b$ or $B_t$ has an active column which is inactive in $T_t$ and $B_b$, we treat this as a new inactive interval that needs to be maintained within $M'$.  The weights from $T_b$ and $B_t$ are added to that in the column for $M$.  If inactive intervals of $M'$ are then adjacent to each other, they are merged, in $O(1)$ time each.  
This completes the recursive step for case (1).  

In the base case when slabs $T$ and $B$ are single rows (at depth $O(\log r)$), the range maximum is restricted to use their active columns.  We sum weights on active columns in $T$, $B$, and $M$.  Then also considering the inactive intervals on $M$, invoke the interval merging procedure~\cite{Barbay2014} to find the maximal range, in time proportional to the number of active intervals, in $O(1 / (2^{\log r} \eta) = O(1 / (r \eta)) $ time.

The cost of recursing in any case is also proportional to the number of active columns since this bounds the number of potential merges, and the time it takes to scan the linked lists of active columns to detect where the merging is needed.  At level $i$ this is bounded by $\ell_i = \min\{r, 1 / (\eta 2^i)\} \le O(1 / (\eta 2^i)) $.  
 
At each level $i$ there are $4^i$ recursive sub instances and at most $O(1/ (2^i \eta))$ active columns, and therefore merging takes $Z_i = 4^i O(1 / (2^i \eta)) = 2^i O(1 / \eta)$ time.  The cost is asymptotically dominated by the last level, which takes time $ 2^{ \log_2 r} O(1/ \eta) = O( r / \eta)$.  
\end{proof}
Letting $\eta = \eps / (\log r) = O(1 / (r \log r))$ (since $r = O(1 / \eps)$) as it is in Lemma \ref{lem:compression} we have a bound of $\c{T}_2(r) = O(r^2 \log r )$.
We can solve the first recurrence of 
$\c{T}_1(r) 
= 
2 \c{T}_1(r/2) + \c{T}_2(r) 
= 
2 \c{T}_1(r/2) + O(r^2 \log r) 
= 
O(r^2 \log r)
$.
Using $r = O(1/\eps)$ this bounds the overall runtime of finding $\max_{R \in (S,{\c{R}_d}_{\mid G})} \Phi(R)$ as $O(\frac{1}{\eps^2} \log \frac{1}{\eps})$.  

\begin{theorem}
\label{thm:rect-linear}
Consider $(X,\c{R}_2)$ with $|X| =m$ and 
$A^* = \arg\max_{A \in \c{R}_2} \Phi(A)$.  With probability at least $1-\delta$, in time $O(m + \frac{1}{\eps^2} \log \frac{1}{\eps} \log \frac{1}{\delta})$, we can find a range $\hat A_\eps$ so 
$
|\Phi(A^*) - \Phi(\hat A_\eps)| \leq \eps.
$
\end{theorem}

In the Appendix \ref{app:loglog-linear}, we reduce this time to $O(m + \frac{1}{\eps^2} \log \log \frac{1}{\eps} \log \frac{1}{\delta})$.  For $(X,\c{R}_d)$ and $d$ constant, the runtime increases to $O(m + \frac{1}{\eps^{2d-2}} + \frac{1}{\eps^2} \log \log \frac{1}{\eps} \log \frac{1}{\delta})$.

\subparagraph*{Conditional lower bound.}
\label{sec:rect-LB}

Backurs \etal~\cite{Backurs16} recently showed $\Omega(m^2)$ time is required to solve for $A^* = \arg \max_{A \in (X, \c{R}_2)} \Phi(A)$, assuming that all pairs shortest path (APSP) requires cubic time.  We can show this implies that our algorithm is nearly tight.  If we set $\eps = 1/4m$ then if any algorithm could find an $\hat A_\eps$ such that $\Phi(\hat A_\eps) \geq \Phi(A^*) - \eps$, then it would imply that $|\mu_R(A^*) - \mu_B(A^*)| - |\mu_R(\hat A) - \mu_B(\hat A)| \leq \eps$.  And hence the difference in counts of points in each pair $\mu_R$ and $\mu_B$ is off by at most $2\eps m = 2(1/4m)m = 1/2$.  Thus it must be the optimal solution.  If this can run in $o(m + 1/\eps^2)$ time, it implies an $o(m^2)$ algorithm, which implies a subcubic algorithm for APSP, which is believed impossible. 

\begin{theorem}
\label{thm:rect-LB}
For $(X,\c{R}_2)$ with $|X| =m$, and $A^* = \arg\max_{A \in \c{R}_2} \Phi(A)$.  It takes $\Omega(m + \frac{1}{\eps^2})$ time to find a range $\hat A_\eps$ so that 
$
|\Phi(A^*) - \Phi(\hat A_\eps)| \leq \eps,
$
assuming APSP takes $\Omega(n^{3})$ time.  
\end{theorem}

\section{Statistical Discrepancy Function Approximation}
\label{sec:fxn-apx}

In this section we address approximating $\max_{A \in (X, \c{A})} \Phi(A)$ when it is a more general function of $\mu_R(A)$, and $\mu_B(A)$.  Rewrite $\Phi(A) = \phi(\mu_R(A), \mu_B(A))$, and in this section it will be more convenient to discuss $\phi(r,b)$ where $r = \mu_R(A)$ and $b= \mu_B(A)$.  

We say $\phi$ is \emph{$(\tau, \gamma)$-linear} if it can be represented with up to $\eps$-error as the upper envelope of $\gamma$ functions of slope at most $\tau$.  We can then simply maximize each function individually, and return the maximum overall score.  When $\gamma$ and $\tau$ are constant (as with $\phi(r,b)=|r-b|$), we simply say the function is \emph{linear}. 

First observe that Theorem \ref{theorem:general-range}, Theorem \ref{thm:halfspace}, and Theorem \ref{thm:rect-gen} simply evaluate $\Phi(A)$, so if this can be done in constant time, and the slope $\tau$ is constant, then these results automatically hold.  
However, Theorem \ref{thm:rect-linear} requires the linearity property.  

For the spatial scan statistic application, the most common function~\cite{Kul97} is defined
$
\phi_K(r,b) = r \ln \frac{r}{b} + (1-r) \ln \frac{1-r}{1-b}, 
$
and is non-linear.  We define a more general class of \emph{statistical discrepancy functions} (\textsc{sdf}), which includes $\phi_K$.  Such $\phi$ have domain $r,b \in [0,1]$, $\phi(r,b) = 0$ when $r=b$ and this is its minimum, and $\phi(r,b)$ is convex on $(0,1)^2$.  
Moreover, for these functions, it suffices too consider a range $[\xi,1-\xi]^2$ for small constant $\xi$ (c.f. \cite{APV06,AMPVZ06,SSSS}), and
that in this range $\phi$ is $\tau$-Lipschitz where $\tau$ is a constant depending only $\xi$.

Agarwal \etal~\cite{APV06} approximated such functions by considering $O(\frac{1}{\eps} \log \frac{1}{\eps})$ linear functions, each tangent to $\phi$, so their upper envelope $\tilde \phi$  satisfied $\max_{(r,b) \in [\xi, 1-\xi]^2} |\phi(r,b) - \tilde \phi(r,b)| \leq \eps$.  

We will construct an approximation of $\phi$ with linear functions with a very different approach.  Unlike the previous approach which only considers the function $\phi$, our approach adapts the set of linear functions to the function $\phi$ \emph{and data} $(X, \c{A})$.  It uses $O(1/\sqrt{\eps})$ linear functions.

\subparagraph*{Function approximation.}  

Consider the distinct ranges in $(X, \c{A})$; each range $A$ corresponds to a point $p_A = (\mu_R(A), \mu_B(A))$.  Let $P = \{p_A \mid A \in (X, \c{A})\}$ be this set of points.  Then $p_{A^*}$, must lie on $\mathsf{CH}(P)$, the convex hull of $P$, where $A^* = \arg \max_{A \in (X, \c{A})} \Phi(A)$.

Moreover, each point $p$ on $\mathsf{CH}(P)$ maximizes some linear function, $f(r,b) = \alpha r + \beta b$. If $p = \arg \max_{p' \in P} f(r_p, b_p)$, then it also maximizes $f_c(r,b) = (\alpha/c) r + (\beta/c) b$ for any $c > 0$.  We can therefore restrict our attention (by implicit choice of $c$) to only functions with $\alpha^2 + \beta^2 = 1$.  These functions correspond to a dot product $\langle (\alpha, \beta), (r,b) \rangle$ and are maximized by points on $\mathsf{CH}(P)$ where $(\alpha, \beta)$ is between two adjacent normals on the boundary of $\mathsf{CH}(P)$.  

To further simplify, we now parameterize these functions by an angle $\theta = \arccos(-\alpha)$ (where still $\alpha^2 + \beta^2 =1$).  
We focus on $\theta \in [0,\pi/2]$ as we can always repeat the procedure on the other 3 quadrants.

Now let $f^*_\theta $ be any linear function such that $p_{A^*} = \arg \max_{p \in P} f^*_\theta(p)$ is maximized by the point $p_{A^*}$ corresponding to the optimal range $A^*$.

\begin{lemma}
\label{lem:triangle}
Consider $p_1 = \arg \max_{p \in P} f_{\theta_1}(p)$ and 
$p_2 = \arg \max_{p \in P} f_{\theta_2}(p)$ so that 
$p_{A^*} = \arg \max_{p \in P} f^*_\theta(p)$ and 
$\theta_1 \leq \theta \leq \theta_2$.  
Then 
$
\phi(p_{A^*}) \leq \max\{ \phi(p_i), \phi(p_j) \} + \tau \cdot \frac{\|p_1 - p_2\|}{2} \tan(\frac{\theta_2 - \theta_1}{2}).
$
\end{lemma}
\begin{proof}
Define a triangle through points $p_1$, $p_2$, and a point $p_3$.  The point $p_3$ is defined at the intersections of the normals to $f_{\theta_1}$ at $p_1$ and to $f_{\theta_2}$ at $p_2$.  
We refer to ``above'' in the normal direction of the edge between $p_1$ and $p_2$, and in the direction of $p_3$.  

First we show that $p_{A^*}$ must be inside the triangle.  If it is above the edge connecting $p_1$ and $p_3$, then it would be $\arg \max_{p \in P} f_{\theta_1}(p)$.  Similarly it cannot be above the edge connecting $p_2$ and $p_3$.  Also, it must be above the edge connecting $p_1$ and $p_2$, since otherwise by convexity $\max(\phi(p_1), \phi(p_2)) > \phi(p_{A^*})$ and one of $p_1$ or $p_2$ would maximize $f^*_\theta$.

We say the height of the triangle $h$ is defined as the distance from $p_3$ to $q_3$, where $q_3$ is the closest point on the edge through $p_1$ and $p_2$.

Let $\angle_1$ be the internal triangle angle at $p_1$, and $\angle_2$ at $p_2$.  Then $(\theta_2 -\theta_1) = \angle_1 + \angle_2$. 
Now $h = ||p_1 - q_3|| \tan(\angle_1) = ||p_2 - q_3|| \tan (\angle_2)$ which, fixing $\|p_1 - p_2\|$, is maximized when $\angle_1 = \angle_2 = \frac{(\theta_2 - \theta_1)}{2}.$
 Summing $h \le ||p_1 - q_3|| \tan((\theta_2 - \theta_1) / 2)$ and $h \le ||p_2 - q_3|| \tan((\theta_2 - \theta_1) / 2)$ it can 
 be seen that $h \le \frac{1}{2}(||p_1 - q_3|| + ||p_2 - q_3||) \tan((\theta_2 - \theta_1) / 2)=\frac{1}{2}(||p_1 - p_2||) \tan((\theta_2 - \theta_1) / 2)$.
Finally, we argue that $\min\{\phi(p_{A^*}) - \phi(p_1), \phi(p_{A^*}) - \phi(p_2)\} \leq \tau \cdot h$.  
Let $\gamma$ be the iso-curve of $\phi$ at value $\phi(p_{A^*})$.  It must pass above $p_1$ and $p_2$, otherwise they would be the maximum.  It also must pass within a distance of $h$ from either $p_1$ or $p_2$ since $\gamma$ is convex, it contains $p_{A^*}$, and $p_{A^*}$ is within $h$ of the edge between $p_1$ and $p_2$.  
Then the lemma follows since $\phi$ is $\tau$-Lipschitz.  
\end{proof}

\begin{wrapfigure}{r}{0.35\textwidth}
	\vspace{-.2in}
	\centering
	\includegraphics[width=0.33\textwidth]{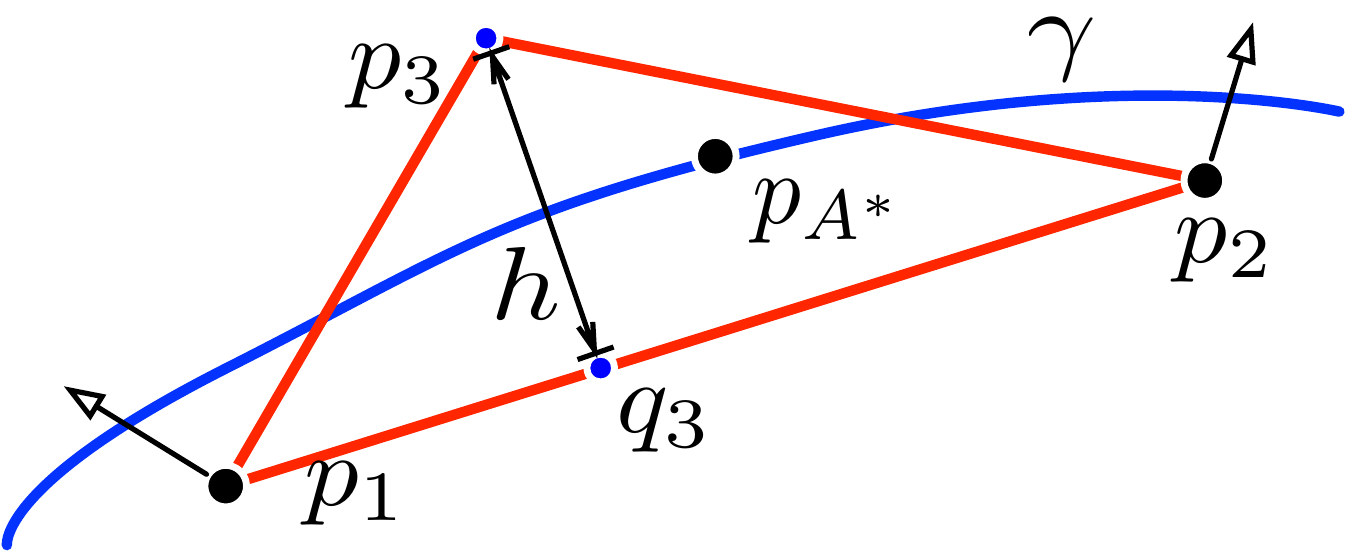}
	\vspace{-.08in}
	\caption{For Lemma \ref{lem:triangle}.}
	\vspace{-.22in}
\end{wrapfigure}
To choose a set of linear functions
we start with two linear functions $f_0$ and $f_{\pi/2}$, whose maximum in $P$ are points $p_1$ and $p'_1$.  These induce a triangle as in the proof of Lemma \ref{lem:triangle}, and $p_{A^*}$ must be in this triangle.  If its height $h = \frac{\|p_1 - p'_1\|}{2} \tan(\frac{\pi}{4}) > \eps/\tau$, then we choose a new function $f_{\pi/4}$ (at the midpoint of the two angles) whose maximum is point $p_2$.  Now recurse on triangles defined by $p_1$ and $p_2$, and by $p_2$ and $p_1'$.  

\begin{lemma}
\label{lem:num-fxn}
The recursive algorithm considers at most $\sqrt{\tau/\eps}$ functions to maximize.  
\end{lemma}
\begin{proof}
Index the points found by the algorithm $\{p_1, p_2, \ldots, p_{k+1}\}$ in the order they appear on the convex hull.  Each consecutive pair $p_i$ and $p_{i+1}$ defines a triangle with height at most $\eps/\tau$.  Let $\ell_i = \|p_i - p_{i+1}\|$ and $\gamma_i = \theta_{i+1} - \theta_i$ where the $p_i$ and $p_{i+1}$ where chosen by maximizing functions $f_{\theta_i}$ and $f_{\theta_{i+1}}$, respectively.  
It follows that $\sum_{i=1}^k \ell_i \leq 2$ and $\sum_{i=1}^k \gamma_i = \pi/2$.  We also have for each triangle that 
$\frac{\eps}{\tau} 
\leq 
\frac{\ell_i}{2} \tan (\frac{\gamma_i}{2}) 
\leq 
\frac{\ell_i}{2} \cdot \frac{2 \gamma_i}{\pi}$.  
Thus for each term we have $\ell_i \geq \frac{\eps \pi}{\tau} \frac{1}{\gamma_i}$, and summing over $k$ terms
$\sum_{i=1}^k \frac{\eps \pi}{\tau} \frac{1}{\gamma_i} \leq \sum_{i=1}^k \ell_i \leq 2$.
Now in the inequality $\frac{2 \tau}{\eps \pi} \geq \sum_{i=1}^k \frac{1}{\gamma_i}$ such that $\sum_{i=1}^k \gamma_i = \pi/2$, then $k$ is the largest when all of the $\gamma_i$ have the same value $\gamma_i = \frac{\pi}{2k}$.  In this case, then 
$\frac{2 \tau}{\eps \pi} 
\geq 
\sum_{i=1}^k \frac{1}{\gamma_i} 
= 
\sum_{i=1}^k \frac{2k}{\pi}
= 
k^2 \frac{2}{\pi}$.  
Solving for $k$ reveals $k \leq \sqrt{\eps/\tau}$.  
\end{proof}

Now we analyze the full algorithm for maximizing a statistical discrepancy function over $(X, \c{R}_d)$ with $\tau$ and $d$ as constants.  We first invoke Lemma \ref{lem:grid-cover} to construct the grid in $O(m + \frac{1}{\eps^2} \log \frac{1}{\eps}\log \frac{1}{\delta} + \frac{1}{\eps^d})$ time.  
We then use Theorem \ref{thm:rect-linear} in $F = O(\frac{1}{\eps^{2d-2}} \log \frac{1}{\eps} )$ time to find the approximate maximum range for any linear function $\Phi'$.  

Then we run the above recursive triangle algorithm repeatedly on the constructed grid, and each function maximization takes $F$ time.  By Lemma \ref{lem:num-fxn} we need to make $O(\sqrt{1/ \eps})$ calls.  And by Lemma \ref{lem:triangle} one of the function calls must find an approximately correct answer.  

\begin{theorem}
\label{thm:rect-stats}
Consider a range space $(X,\c{R}_d)$ with $|X| =m$ and $d$ constant.  For a statistical discrepancy function $\Phi$ with $\tau$ constant and with maximum range $A^* = \arg\max_{A \in \c{R}_d} \Phi(A)$,  then with probability at least $1-\delta$, in time $O(m + \frac{1}{\eps^{2d -1.5}}\log \frac{1}{\eps} + \frac{1}{\eps^2} \log\frac{1}{\eps }\log \frac{1}{\delta})$, we can find a range $\hat A_\eps$ so that 
$
|\Phi(A^*) - \Phi(\hat A_\eps)| \leq \eps.
$
\end{theorem}

\section{Experiments on Rectangles}
\label{sec:exp}
We implemented $5$ rectangle scanning algorithms.  
For baselines, we consider 
(1) Scanning all rectangles without sampling (based on common software for disks~\cite{Kul7.0}) (\SatScan (no sampling)), 
(2) Scanning all rectangles on one random sample~\cite{AMPVZ06} (\SatScan), and
(3) Scanning all rectangles on two random samples $N$ and $S$~\cite{SSSS} (\netScan).  
Then we compare our algorithms which first round to a grid then 
(4) Efficiently enumerate the grid rectangles (\gridScan, Theorem \ref{thm:rect-gen}), or 
(5) Evaluate the maximum grid rectangle in $O(r^3)$ time~\cite{Bentley84} for a linear $\phi$ (\gridScanLin, Section \ref{subsection:linear-algorithm}) and using the linearization for non-linear $\phi$ (Section \ref{sec:fxn-apx}).  
This is the core operation within spatial scan statistics; it is typically run $1000$ times to detect a region \emph{and} determine significance~\cite{Kul97}, therefore scalability of this operation is paramount.  
Solutions with approximate $\phi$ within $\eps$-error retain high statistical power~\cite{SSSS}, so it will be useful to directly compare the runtime performance of these algorithms which allow approximation.

First, fixing a tolerable error at $1\%$ of $\phi(A^*)$, we run each algorithm on $m=1000$ points, for a planted range with $5\%$ of the data, and use $\phi$ as the Kuldorff scan statistic~\cite{Kul97}.  The results are in Table \ref{tbl:exact_results}.  All sampling methods drastically improve over the brute force approach, and using two-level sampling significantly improves over one random sample.  Our method (\gridScanLin) improves over the previous best (\netScan) by a factor of about $3.5$.

\begin{table}[h!]
	\vspace{-.05in}
	\begin{tabular}{ c | c  c   c  c  c }
		\hline
		& \SatScan (no sampling) & \SatScan & \netScan & \gridScan & \gridScanLin 
		\\ 
		 Time (sec) & 5287 & 7.44  & .0279 & .0194 & .0082 \\
		\hline
\end{tabular} 

\caption{\label{tbl:exact_results}
\s{Runtimes on $1000$ points with $1\%$ error, over $20$ trials; roughly $n=19$ and $s=350$.}  }
\vspace{-.08in}
\end{table}
We also compare the time-accuracy trade-off for sampling-based algorithms on $m = 1$ \emph{million} points.  \SatScan without sampling is not tractable at this scale, so is not compared.  
We again plant a random rectangle $A$ overlapping $1\%$ of the data.  Within $A$, points are made red (measured value $1$) at rate $0.08$, and outside at rate $0.01$.  The runtime includes the time to construct the grid, but not time to generate the initial sample -- common to all algorithms.  We calculate $\Phi(A^*) - \Phi(\hat A)$ for the planted $A^*$ and found $\hat A$ regions, using a linear $\phi(m,b) = \frac{1}{\sqrt{2}}(m-b)$ function and the non-linear Kuldorff~\cite{Kul97} $\phi$ function.   
Figure \ref{fig:discr-plot} 
shows a kernel regression trend line (with $1$ std-dev error bars) for $300$ trials with various $n,s$ values, always maintaining $n \approx \sqrt{s}$ as suggested the samping theorems.  
Again \gridScanLin is much faster than \gridScan, which is slightly faster than \netScan, which is significantly faster than \SatScan.  The improvement is more pronounced in the non-linear setting where $\phi$ is steeper; this is perhaps surprisingly even true for \gridScanLin which has an extra $\sqrt{1/\eps}$-factor in runtime in that case due to the multiple linear functions considered.  

\begin{figure}[h!]
\vspace{-.1in}
\includegraphics[width=0.5\textwidth]{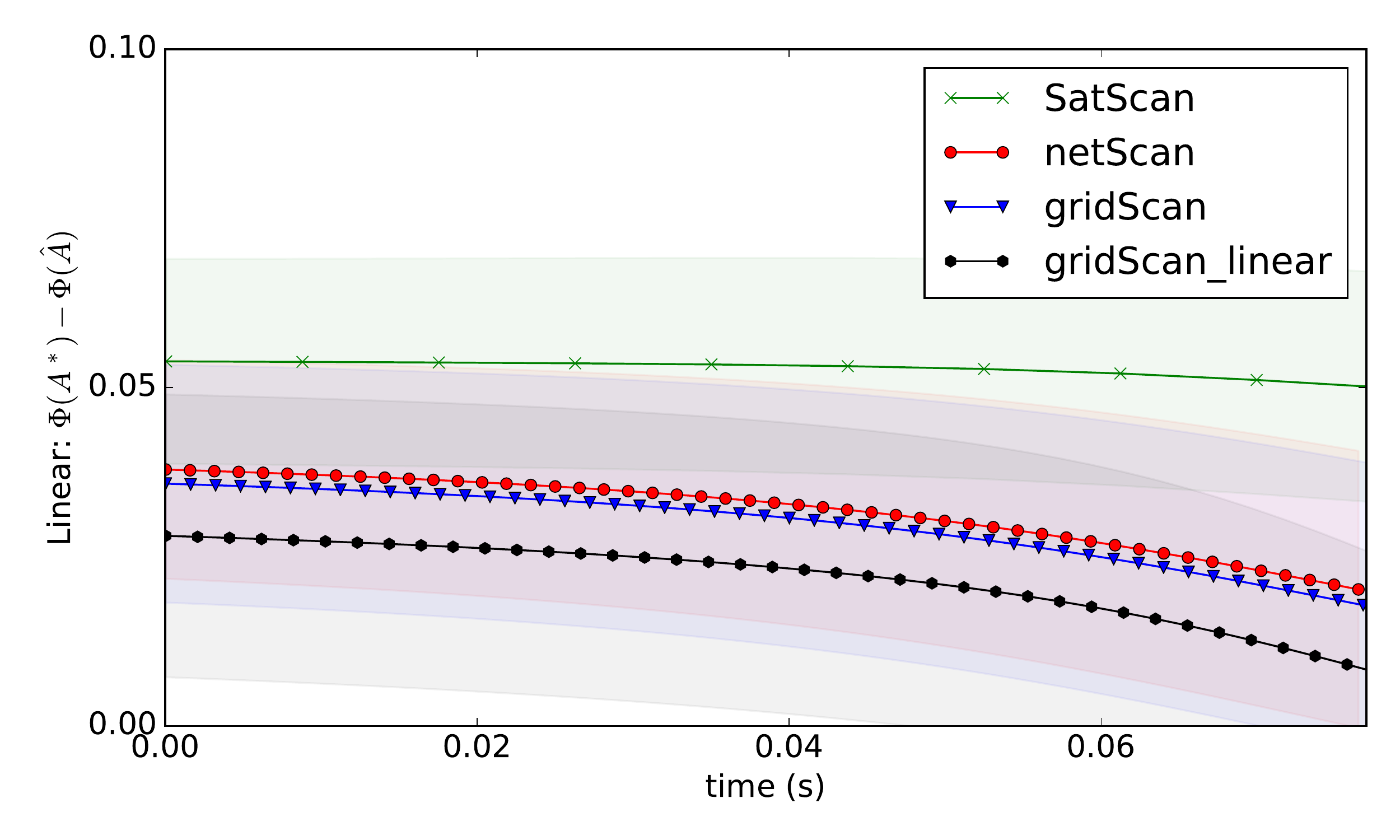}
\includegraphics[width=0.5\textwidth]{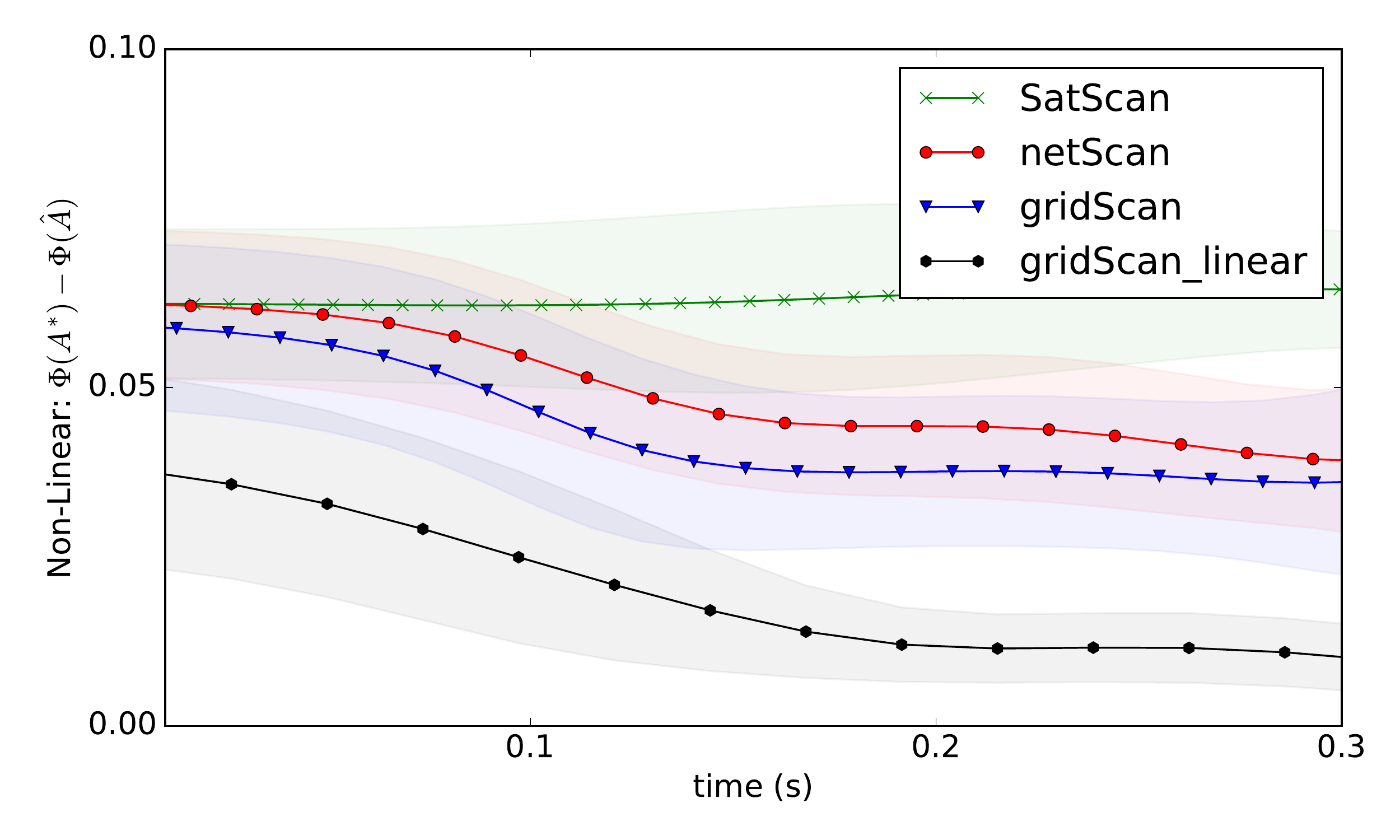}

\vspace{-.2in}
	\caption{\s{Trend of time versus error for on linear (left) and non-linear (right) functions.}}
	\label{fig:discr-plot}
	\vspace{-.15in}
\end{figure}

Ultimately, these plots show that \emph{discrete geometric approaches providing asymptotically efficient algorithms also give significant empirical improvements}, even compared to the ubiquitous and simple random sampling approaches.

\bibliographystyle{plainurl}
\bibliography{bib-disc}

\begin{thebibliography}{10}

\bibitem{AMPVZ06}
Deepak Agarwal, Andrew McGregor, Jeff~M. Phillips, Suresh Venkatasubramanian,
  and Zhengyuan Zhu.
\newblock Spatial scan statistics: Approximations and performance study.
\newblock In {\em KDD}, 2006.

\bibitem{APV06}
Deepak Agarwal, Jeff~M. Phillips, and Suresh Venkatasubramanian.
\newblock The hunting of the bump: On maximizing statistical discrepancy.
\newblock {\em CoRR}, 2005.
\newblock \href {http://arxiv.org/abs/cs/0510004} {\path{arXiv:cs/0510004}}.

\bibitem{Backurs16}
Arturs Backurs, Nishanth Dikkala, and Christos Tzamos.
\newblock Tight hardness results for maximum weight rectangles.
\newblock In {\em ICALP}, 2016.
\newblock URL: \url{http://arxiv.org/abs/1602.05837}.

\bibitem{Barbay2014}
J{\'e}r{\'e}my Barbay, Timothy~M. Chan, Gonzalo Navarro, and Pablo
  P{\'e}rez-Lantero.
\newblock Maximum-weight planar boxes in time (and better).
\newblock {\em Information Processing Letters}, 114(8):437 -- 445, 2014.

\bibitem{Bentley84}
Jon Bentley.
\newblock Programming pearls -- perspective on performance.
\newblock {\em Communications of ACM}, 27:1087--1092, 1984.

\bibitem{Cha01}
Bernard Chazelle.
\newblock {\em The Discrepancy Method}.
\newblock Cambridge, 2000.

\bibitem{DE93}
David Dobkin and David Eppstein.
\newblock Computing the discrepancy.
\newblock In {\em Proceedings 9th Annual Symposium on Computational Geometry},
  1993.

\bibitem{DEM96}
David~P. Dobkin, David Eppstein, and Don~P. Mitchell.
\newblock Computing the discrepancy with applications to supersampling
  patterns.
\newblock {\em ACM Trans. Graph.}, 15(4):354--376, October 1996.

\bibitem{EDEL89}
Herbert Edelsbrunner and Leonidas~J. Guibas.
\newblock Topologically sweeping an arrangement.
\newblock {\em Journal of Computer and System Sciences}, 38(1):165 -- 194,
  1989.

\bibitem{FMMT96}
Takeshi Fukuda, Yasukiko Morimoto, Shinichi Morishita, and Takeshi Tokuyama.
\newblock Data mining using two-dimensional optimized association rules:
  Scheme, algorithms, and visualization.
\newblock {\em SIGMOD Rec.}, 25(2):13--23, June 1996.

\bibitem{Hau95}
David Haussler.
\newblock Sphere packing numbers for subsets of the boolean $n$-cube with
  bounded {Vapnik-Chervonenkis} dimension.
\newblock {\em J. Combinatorial Theory, A}, 69:217--232, 1995.

\bibitem{HKG07}
Lan Huang, Martin Kulldorff, and David Gregorio.
\newblock A spatial scan statistic for survival data.
\newblock {\em BioMetrics}, 63:109--118, 2007.

\bibitem{Kul97}
Martin Kulldorff.
\newblock A spatial scan statistic.
\newblock {\em Communications in Statistics: Theory and Methods},
  26:1481--1496, 1997.

\bibitem{Kul7.0}
Martin Kulldorff.
\newblock {\em SatScan User Guide}.
\newblock http://www.satscan.org/, 7.0 edition, 2006.

\bibitem{Kulldorff2006}
Martin Kulldorff, Lan Huang, Linda Pickle, and Luiz Duczmal.
\newblock An elliptic spatial scan statistic.
\newblock {\em Statistics in medicine}, 25 22:3929--43, 2006.

\bibitem{LLS01}
Yi~Li, Philip~M. Long, and Aravind Srinivasan.
\newblock Improved bounds on the samples complexity of learning.
\newblock {\em J. Comp. and Sys. Sci.}, 62:516--527, 2001.

\bibitem{Lin96}
Ming~C Lin and Dinesh Manocha.
\newblock {\em Applied Computational Geometry. Towards Geometric Engineering:
  Selected Papers}, volume 114.
\newblock Springer Science \& Business Media, 1996.

\bibitem{MP18}
Michael Matheny and Jeff~M. Phillips.
\newblock Practical low-dimensional halfspace range space sampling.
\newblock In {\em European Symposium on Algorithms (arXiv:1804.11307)}, 2018.

\bibitem{SSSS}
Michael Matheny, Raghvendra Singh, Liang Zhang, Kaiqiang Wang, and Jeff~M.
  Phillips.
\newblock Scalable spatial scan statistics through sampling.
\newblock In {\em SIGSPATIAL}, 2016.

\bibitem{Mat99}
Jiri Matou\v{s}ek.
\newblock {\em Geometric Discrepancy}.
\newblock Springer, 1999.

\bibitem{Mat02}
Jiri Matou\v{s}ek.
\newblock {\em Lectures in Discrete Geometry}.
\newblock Springer, 2002.

\bibitem{NM04}
Daniel~B. Neill and Andrew~W. Moore.
\newblock Rapid detection of significant spatial clusters.
\newblock In {\em KDD}, 2004.

\bibitem{Sau72}
Norbert Sauer.
\newblock On the density of families of sets.
\newblock {\em Journal of Combinatorial Theory, Series A}, 13:145--147, 1972.

\bibitem{TAKAOKA2002}
Tadao Takaoka.
\newblock Efficient algorithms for the maximum subarray problem by distance
  matrix multiplication.
\newblock {\em CATS}, 2002.

\bibitem{Tango2005}
Toshiro Tango and Kunihiko Takahashi.
\newblock A flexibly shaped spatial scan statistic for detecting clusters.
\newblock {\em International Journal of Health Geographics}, 4(1):11, May 2005.

\bibitem{VC71}
Vladimir Vapnik and Alexey Chervonenkis.
\newblock On the uniform convergence of relative frequencies of events to their
  probabilities.
\newblock {\em Theo{.} of Prob and App}, 16:264--280, 1971.

\bibitem{walther2010}
Guenther Walther.
\newblock Optimal and fast detection of spatial clusters with scan statistics.
\newblock {\em Ann. Statist.}, 38(2):1010--1033, 04 2010.

\bibitem{WSJRG09}
Mingxi Wu, Xiuyao Song, Chris Jermaine, Sanjay Ranka, and John Gums.
\newblock A {LRT} framework for fast spatial anomaly detection.
\newblock In {\em KDD}, 2009.

\end{thebibliography}
\normalsize

\newpage

\appendix

\section{Removing the $\log \frac{1}{\eps}$ in Linear Rectangle Construction} 
\label{app:loglog-linear}
In this section we make three small modifications to the linear rectangle algorithm to reduce the $\log \frac{1}{\eps}$ dependency to a  $\log \log \frac{1}{\eps}$ factor for the $2$-dimensional case (in higher dimensions other factors dominate). 
These steps are:
(1) a fast construction of a slightly smaller $\eps$-approximation,
(2) constructing a separate grid for each slab spanning rectangle problem instance,
(3) a method to maintain active and inactive columns with respect to a single separating line.
These changes modify the structure of the rectangle scanning algorithm. In Section \ref{sec:rect} the structure of the algorithm is: sample, process the sample into a single grid, construct a single compressed slab tree, and then evaluate many slab spanning rectangle problems on 
this tree. In this section the structure changes to: compute a sample, process the sample into a dyadic set of partitions, define a grid on each partition, construct a compressed slab tree for each grid, and then evaluate each slab spanning rectangle problem on 
its own tree.

\subsection{(P1) Smaller Sample Construction}
\label{app:ssc}

In Lemma \ref{lem:grid-cover} the construction of an $\eps$-cover on an input of size $s$ takes 
$O(s \log s + \frac{1}{\eps^d})$ time. This construction creates a bottleneck when $s = O(\frac{1}{\eps^2})$ as in standard random sampling. Instead
we will take a larger random sample, $S'$ of size $s'$ (specified in proof), and use this to generate a smaller sample, $S$ of size $s = O(\frac{1}{\eps^2} / \log^2 \frac{1}{\eps})$.

\begin{lemma}
	\label{lemma:grid-sample}
	The construction of an $\eps$-approximation for $(X, \c{R}_d)$ of size $O(\frac{1}{\eps^2 \log^2 \frac{1}{\eps}})$  can be done in $O(m + \frac{1}{\eps^2} \log \log \frac{1}{\eps} \log \frac{1}{\delta})$ time for constant $d > 1$.  
\end{lemma}
\begin{proof}
	We first generate an $\frac{\eps_1}{2d}$-approximation $S'$ using random sampling (in time $O(m)$ and size of $s' = O(\frac{1}{\eps_1^2} \log \frac{1}{\delta})$) and then partition it into $z = \log^{cd} \frac{1}{\eps_2}$ parts by recursively halving at the median on successive coordinates. This is the same construction as used to construct $k$d-trees and it can be noted that a tree with $z=\log^{cd}\frac{1}{\eps_2}$ leaves, for constant $c$, will have height $(\log \log^{cd} \frac{1}{\eps_2}) = (\log (cd \log \frac{1}{\eps_2}))$ and the partitioning will take $O(s' \log \log \frac{1}{\eps_2})$ time. 
	
	For each partition we randomly sample $p$ points for a total of $pz=s$ points.
	Given $z$ partitions any rectangle will cross at most $O(z^{1 - \frac{1}{d}})$ partitions. 
	Fixing a single query rectangle we can apply a Hoeffding bound on the error contribution over the partitions crossed. 
	For each point in $S$ contained in one of these $O(z^{1-\frac{1}{d}})$ crossed partition we can associated with it a random variable $X_i$ where $X_i$ is $1$ if the point is inside the query rectangle and $0$ if not. Over the at most $k = p c' z^{1 - \frac{1}{d}}$ such points, for some constant $c'$, contained in the crossed partitions, we apply Hoeffding's inequality on $X = \sum_{i = 1}^{k} X_i$, with chosen bound $\delta$.  $\E[X]$ is the expected number of points in the crossed partitions which are contained in the rectangle.   
	\[
	\Pr(|X - \E[X]| > p z \eps_2) 
	\le 
	2 \exp \left (- \frac{2 (p z \eps_2)^2}{p c' z^{1 - \frac{1}{d}}}  \right ) 
	= 
	2 \exp \left ( -\frac{2 p z^{1 + \frac{1}{d}}  \eps_2^2}{c'}  \right ) 
	\le 
	\delta.  
	\]
	Rearranging this gives 
	$\frac{c'}{2\eps_2^2 \cdot z^{\frac{1}{d}}} \log \frac{2}{\delta} \le p z = s$.  
	Since $z = \log^{cd}\frac{1}{\eps_2}$ then 
	$\frac{c'}{2\eps_2^2 \cdot \log^c \frac{1}{\eps_2}} \log \frac{2}{\delta} \le  s$. 
	Setting $\delta = o(\eps_1^{2d})$ allows a union bound over all rectangular regions on $S'$ with $s = O(\frac{1}{\eps_2^2 }\frac{\log \frac{1}{\eps_1}}{\log^{c} \frac{1}{\eps_2}})$. Since $S$ is an $\eps_2$-sample over $S'$ by setting $\eps_1 = \eps_2 = \frac{\eps}{2}$ the additive error is $\eps_1 + \eps_2 = \eps$. For appropriate $c$ the size of $|S| = O(\frac{1}{\eps^2 \log^2 \frac{1}{\eps}})$. 
\end{proof}

\subsection{(P2) Parametrized $\eps$-Covers for Grids}

Next we reexamine the grid construction in Section \ref{sec:grid}, and parametrize it so that if we only want an $\eps$-cover on the set of rectangles that are constrained to some region $R$ where $|X \cap R| = |X|/2^i$ we only need $\ell_i = r/2^i = O(1/(\eps 2^i))$ cells on each side. Using Lemma \ref{lem:grid-cover} on $|X \cap R|$ with  $\eps 2^i$ we construct a grid $G$ in time $O(m/2^i \log m + \frac{1}{(\eps 2^i)^2} )$  which 
induces a $\eps 2^i$-cover $(X \cap R, {\c{R}_2}_{\mid G})$ on $(X \cap R, {\c{R}_2}_{\mid X \cap R})$. 
Since  ${\c{R}_2}_{\mid G}$ is an $\eps /2^i$-cover $(X \cap R, {\c{R}_2}_{\mid X \cap R})$ then for any $A' \in  {\c{R}_2}_{\mid X \cap R}$ there exists an $A \in {\c{R}_2}_{\mid G}$ such that $|A \triangle A'| \le 2^i \eps |X \cap R| = 2^i \eps |X|/2^i = \eps|X|$ and therefore ${\c{R}_2}_{\mid G}$ is an $\eps$-cover on $(X, {\c{R}_d}_{\mid X \cap R}))$.

Using Lemma \ref{lemma:grid-sample} we can construct an $S$ of size $s = O(\frac{1}{\eps^2} / \log^2 \frac{1}{\eps})$. Setting $m = s/ 2^i$ then we 
arrive at the runtime
$O((s/2^i) \log \frac{1}{\eps} + \frac{1}{\eps^2 4^i})$ and state the following result. 

\begin{lemma}
	\label{lemma:grid-cover3}
	Consider a point set $X$ of size $s$ with $s = O(\frac{1}{\eps^2} / \log^2 \frac{1}{\eps})$ with a sub-set of $X$ constrained to a region $R$,
	such that $|R \cap X| = s/2^i$.
	We can construct a grid $G$ in time $O(\frac{1}{\eps^2 2^i \log \frac{1}{\eps}})$, such that it induces an $\eps$-cover $(X, {\c{R}_2}_{\mid G})$ of $(X, {\c{R}_2}_{\mid R \cap X})$.  
\end{lemma}

\subsection{(P3) Improving Tree and Slab Approximation}

Next we improve the tree and slab approximation parametrized to an $\ell \times \ell$ grid $G$ containing $\gamma$ points.
We will only consider grid-aligned rectangles $A = [x_1, x_2] \times [y_1, y_2]$ where $y_1$ is in the lower half of the grid, and $y_2$ is in the upper half of the grid. Call this a \emph{grid-spanning rectangle}. Let $y_m$ represent the $y$-coordinate of the middle coordinate of this
grid.  In this setting we will guarantee that we have at most $O(\eps \gamma)$ error in the tree on each half with a more clever way of constructing active columns.  We focus only on the lower half of the grid, the upper half is symmetric.  

\begin{wrapfigure}{R}{0.3\textwidth}
	\vspace{-.4in}
	\includegraphics[width=0.3\textwidth]{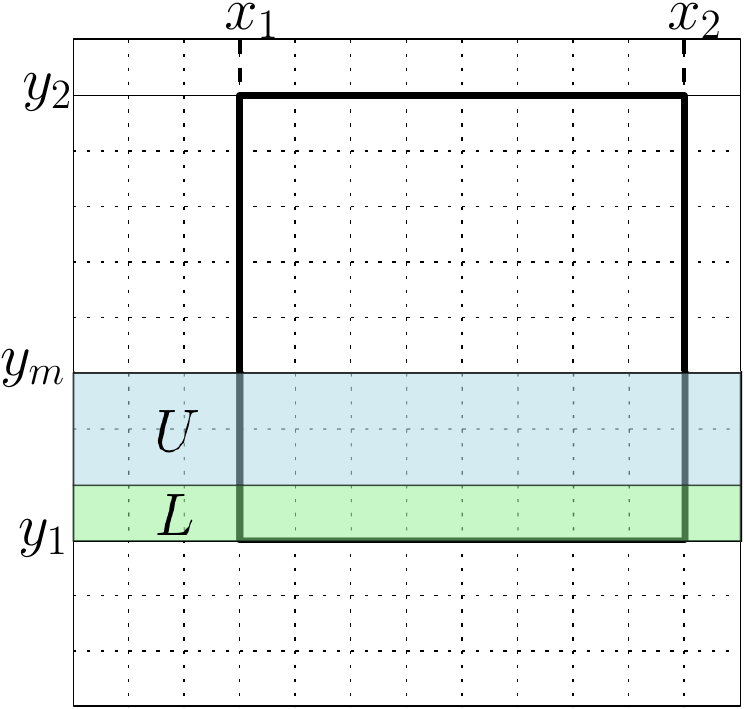}
	\vspace{-.4in}
	\label{fig:TreeSlabApp}
\end{wrapfigure}

We can decompose the $y$-coordinate range $[y_1, y_m]$ into at most $\log_2 r$ slabs. Recall, that before we would compress the weights in a slab by scanning cells from left to right within a slab, and only retaining the cumulative weight when it would exceeds $\eta \gamma$. This accumulates an error of $O(\eta \gamma \log r)$ error across these slabs, and therefore setting $\eta = O(\eps / \log r)$ gave $\eps \gamma$ error. We refer to this 
processes of scanning slabs and retaining the cumulative sums when they exceed $\eta \gamma$ as $\eta$-compression. By modifying this scheme to
compress with respect to already compressed slabs we can set $\eta = O(\eps)$ which will improve the runtime in Lemma \ref{lem:rect-strip}. 
Consider compressing $L$, the lowest most slab of $[y_1, y_m]$.  
Let $U$ represent $[y_1,y_m] \setminus L$. Note that $U$ could be the union of many disjoint slabs.
Let $w_{U,i}$ be the total weight in column $i$ of $G$ within $U$, and likewise $w_{L,i}$ be the total weights in column $i$ in $L$. Assume we have already compressed all slabs that make up $U$ into a compressed slab  $\hat{U}$.
For consistent notation $\hat w_{U,i}$ will be the compressed total weight in column $i$ of $G$ within $U$.
  
Now we want to approximate the weights $w_{L,i} + w_{U, i}$ with approximate weight $\hat w_{L,i} + \hat w_{U, i}$, so that the sum of any lower half rectangle has at most $\eta = \eps$ additive error. The key idea now is that we will apply $\eta$-compression on  $w_{L,i} + w_{U,i} - \hat w_{U,i}$ (all of which we will know inductively) to generate $\hat w_{L, i}$.  We denote each column (and $x$ value) within a slab where it has retained a non-zero value as \emph{active}, all other columns are \emph{inactive}. 

First denote the summed weight $\sum_j w_{U,j} = W_{U}$ (and likewise $W_L$, $W_U$, and $\hat W_U$ are the total summed weight over $w_{L,i}$,  $w_{U,i}$, and $\hat w_{U,i}$).  Then we first show that $W_L = \hat W_L$.  Since the compression conserves the cumulative sum of the compressed rows, $\hat W_{L} = W_L + W_U - \hat W_U$. Assume by induction that 
$W_U = \hat W_U$, the base case holds since the initially compressed slab in $U$ has weight equal to its compressed weight.  
Then by induction $\hat W_L = W_L + (W_U - \hat W_U) = W_L$, as desired.  
Now, since $\hat W_L = W_L$, we can use its original weight $W_L$ to bound the number of active columns.  In particular, a compressed slab $L$ at level $i$ in the tree will contain $W_L / (\eta \gamma) = 1 / ( 2^{i} \eps)$ active columns.

To analyze error note that we will over count on the left boundary of the rectangle and undercount on the right boundary of the rectangle, and therefore if we can bound the error along one side the total error (sum of the error on the left side plus the error on the right side) will be less than this. We therefore focus on a left sided rectangle $[x_1, \infty] \times [y_1, y_m]$, for \emph{any} choice of $x_1$.   Compressing a slab $L$ during $\eta$-compression involves computing a cumulative sum until it reaches some values equal to or less than $\eta \gamma$ at which point the cumulative count is set back to $0$ again.
Let $x'_1$ be the first index where this sum was reset to $0$ again before $x_1$ that lies in $L$.  The sum from $x_1'$ to the end of the row is 
$\sum_{k = x'_1}^{\ell} \hat w_{L,k} =  \sum_{k = x'_1}^{\ell} w_{L, k} + w_{U,k} - \hat w_{U,k}$ (since these are perfectly compressed) and $\sum_{k = x'_1}^{x_1-1} \hat w_{L, k} = 0$ (since there is no count, by definition).  Now we can consider the compression error on rectangle 
$[x_1, \infty] \times [y_1, y_m]$, written as error from compressing $w_L$ to $\hat w_L$ and from compressing all of $w_U$ to $\hat w_U$ as
\begin{align*}
\MoveEqLeft {\sum_{k = x_1}^{\ell}   (\hat w_{L,k} - w_{L,k}) + (\hat w_{U,k} - w_{U,k}) 
= 
\sum_{k = x_1}^{\ell}   \hat w_{L,k} - (w_{L,k} + w_{U,k} - \hat w_{U,k}) }  
\\ &=
(\sum_{k = x'_1}^{\ell}  \hat w_{L,k} - (w_{L,k} + w_{U,k} - \hat w_{U,k})) - \sum_{x'_1 \le k < x_1} \hat w_{L,k} - (w_{L,k} + w_{U,k} - \hat w_{U,k}) 
\\ &= 
0 + \sum_{x'_1 \le k < x_1} (w_{L,k} + w_{U,k} - \hat w_{U,k}) - \hat w_{L,k}
\\ &= 
\sum_{x'_1 \le k < x_1} (w_{L,k} + w_{U,k} -\hat w_{U,k} ) \le \eta \gamma.
\end{align*}

The additive error on $[x_1, x_2] \times [y_1, y_m]$ is therefore less than or equal to $\eta \gamma$.
The scanning to perform this compression takes $O(\ell)$ time for one slab.  And thus for all $ \sum_{i=1}^{\log \ell} 2^i = O(\ell)$ slabs this takes $O(\ell^2)$ time.

\begin{lemma} 
	\label{lem:compression2}
	For a size $\ell \times \ell$ grid $G$ containing $\gamma$ points, in $O(\ell^2)$ time, we can compress $G$ so a dyadic slab at level $i$ has at most $1 / ( \eta 2^i)$ active columns, and any grid-spanning rectangle has at most $O(\eta \gamma) = O(\eps \gamma)$ error.  
\end{lemma}

Plugging in the new value $\eta = O(\eps) = O(1/\ell)$  into Lemma \ref{lem:rect-strip} we can restate the time to compute the maximum grid-spanning rectangle. 
\begin{lemma}
	\label{lem:rect-strip2}
	The approximate Strip-constrained grid search problem of size $\ell$ can be solved in time $O(\ell^2)$ such that the maximum rectangle can be found with $\eps \gamma$ error.
\end{lemma}

\subsection{Putting it all Together}

Now we start with a point set $X$ and using \textbf{(P1)} construct an $\eps$-sample $S$ of 
size $s = O(\frac{1}{\eps^2\log\frac{1}{\eps}})$ on it in $O(m + \frac{1}{\eps^2} \log \log \frac{1}{\eps} \log \frac{1}{\delta})$ time.  

On the full set, we call Lemma \ref{lemma:grid-cover3} with $i=0$ to create an $\eps$-cover, and Lemma \ref{lem:compression2} with $\ell = r = O(1/\eps)$.  Now we can call the Strip-constrained grid search on this compressed representation in $O(\ell^2) = O(r^2) = O(1/\eps^2)$ time, and if the maximal rectangle is grid-spanning, we will find a good approximation of it. If the rectangle is not grid-spanning then it must lie above or below the middle 
separating line, so we recurse on the upper and lower half of $S$.
This leads to a recurrence on the sample $S$ of size $s$ so the total time will be
$\c{T}_1(s) = 2\c{T}_1(s/2) + \c{T}_2(s)$.
Here the function $\c{T}_2(s)$ is the time needed for preprocessing and solving the grid-spanning rectangle problem. At depth $i$ in the recurrence we compute the grid (Lemma \ref{lemma:grid-cover3} with parameter $i$), and find the grid spanning rectangle (Lemma \ref{lem:rect-strip2} with $\ell= r / 2^i$), which implies that $\c{T}_2(s / 2^i) = O(\frac{1}{\eps^2 2^i \log \frac{1}{\eps}} + \frac{1}{\eps^2 4^i} + \frac{1}{\eps^2 4^i} )$.  Solving the recurrence,  $\c{T}_1(s)$, results in 
$\c{T}_1(s) = \sum_{i = 0}^{\log s}  O(\frac{1}{\eps^2 4^i} + \frac{1}{\eps^2 2^i\log \frac{1}{\eps}} ) = O(\frac{1}{\eps^2})$. The dominant term is then the preprocessing time to generate $S$ and the total algorithm runs in time $O(m + \frac{1}{\eps^2}\log \log \frac{1}{\eps} \log \frac{1}{\delta})$.

\begin{theorem}
	\label{thm:rect-linear2}
	Consider a range space $(X,\c{R}_2)$ with $|X| =m$.  For a linear function $\Phi$ with maximum range $A^* = \arg\max_{A \in \c{R}_2} \Phi(A)$, with probability at least $1-\delta$, in time $O(m + \frac{1}{\eps^2} \log \log \frac{1}{\eps}\log \frac{1}{\delta})$, we can find a range $\hat A_\eps$ so that 
	$
	|\Phi(A^*) - \Phi(\hat A_\eps)| \leq \eps.
	$
\end{theorem}

Once $S$ and the dyadic decomposition of grids are built they can be reused for multiple linear function executions. The 
time to construct $S$ and the decomposition is $O(m + \frac{1}{\eps^2} \log \log \frac{1}{\eps} \log \frac{1}{\delta})$, the time to find the maximum rectangle for each linear function execution is $O(\frac{1}{\eps^2})$ and by Lemma \ref{lem:num-fxn} the number 
of these executions is $O(\frac{1}{\sqrt{\eps}})$.
\begin{theorem}
	\label{thm:rect-stats2}
	Consider a range space $(X,\c{R}_2)$ with $|X| =m$.  For a statistical discrepancy function $\Phi$ with $\tau$ constant and with maximum range $A^* = \arg\max_{A \in \c{R}_d} \Phi(A)$,  then with probability at least $1-\delta$, in time $O(m + \frac{1}{\eps^{2.5}} + \frac{1}{\eps^2} \log \log \frac{1}{\eps} \log \frac{1}{\delta})$, we can find a range $\hat A_\eps$ so that 
	$
	|\Phi(A^*) - \Phi(\hat A_\eps)| \leq \eps.
	$
\end{theorem}

To apply this to higher dimensions we can apply the same machinery, but scan by fixing all, but two dimensions. This gives a runtime of $O(m + \frac{1}{\eps^{2d -2}} + \frac{1}{\eps^2} \log \log \frac{1}{\eps} \log \frac{1}{\delta})$ for the linear version of our problem and $O(m + \frac{1}{\eps^{2d -2 + 1/2}} + \frac{1}{\eps^2} \log \log\frac{1}{\eps} \log \frac{1}{\delta})$ for the statistical discrepancy version.

\end{document}